\documentclass[smallextended,runningheads]{svjour3}
\journalname{Annals of the Institute of Statistical Mathematics}
\smartqed  

\usepackage{amsmath,amssymb, multirow, tabularx, booktabs, amsfonts,  graphicx, geometry}

\usepackage[colorlinks=true, allcolors=blue]{hyperref}

\newtheorem{assumption}{}

\begin{document}
\title{Robust Variable Selection in High-dimensional Nonparametric Additive Model}
%
%
%

\author{Suneel Babu Chatla 
 and Abhijit Mandal}

\institute{Suneel Babu Chatla \at
              University of Texas at El Paso \\
              \email{sbchatla@utep.edu}           
           \and
           Abhijit Mandal \at
             University of Texas at El Paso \\
             \email{amandal@utep.edu}
}

\maketitle

\begin{abstract}
Additive models belong to the class of structured nonparametric regression models that do not suffer from the curse of dimensionality. Finding the additive components that are nonzero when the true model is assumed to be sparse is an important problem, and it is well studied in the literature. The majority of the existing methods focused on using the $L_2$ loss function, which is sensitive to outliers in the data. We propose a new variable selection method for additive models that is robust to outliers in the data. The proposed method employs a nonconcave penalty for variable selection and considers the framework of B-splines and density power divergence loss function for estimation. The loss function produces an M-estimator that down weights the effect outliers. Our asymptotic results are derived under the sub-Weibull assumption, which allows the error distribution to have an exponentially heavy tail. Under regularity conditions, we show that the proposed method achieves the optimal convergence rate. In addition, our results include the convergence rates for sub-Gaussian and sub-Exponential distributions as special cases. We numerically validate our theoretical findings using simulations and real data analysis.



\end{abstract}

\section{Introduction}

Additive models introduced by \cite{friedman1981projection} and \cite{stone1985additive} and popularized by \cite{hastie1986gam} belong to the class of structured nonparametric regression models. They are known for their flexibility and do not suffer from the curse of dimensionality, which arises in the estimation of multivariate nonparametric functions. The pioneering work of \cite{stone1985additive} shows that the additive model can be estimated with the same optimal rate of convergence for univariate functions. Let $(Y_i, X_{i1}, \ldots, X_{ip})$, $i=1,2,\ldots,n$, be $n$ independent and identically distributed (i.i.d.) copies of $(Y, X_1, \ldots, X_p)$, where $Y$ is a scalar response and $\{ X_j, j=1,\ldots, p\}$ are covariates. We consider the following model:
\begin{align}
    Y_i &= \mu + \sum_{j=1}^p g_j(X_{ij}) + \epsilon_i, \qquad i=1,\ldots,n,
    \label{eqn:model-def}
\end{align}
where $\epsilon_i$ is the random error with mean zero and finite variance, $\sigma^2$, and $g_j(\cdot)$'s are unknown functions with $\mathbb{E}[g_j(X_{ij})]=0$, $j=1,2,\ldots,p$. When $p$ is much smaller than $n$, the popular methods for estimation of (\ref{eqn:model-def}) include backfitting \cite{buja1989linear,hastie1986gam}, penalized splines \cite{wood2006generalized}, smooth backfitting \cite{mammen1999existence}, and marginal integration \cite{linton1995kernel}.

Since the last decade, there has been a growing interest in high-dimensional additive models. When $p$ is large, it is natural to assume sparsity in the component functions so that only a few selected component functions are used to fit the additive model. Early contributions primarily focused on fitting splines for component functions and achieving sparsity through group penalties; see \cite{lin2006component,meier2009high,ravikumar2009sparse} and \cite{huang2010variable} for contributions related to group LASSO type of penalties and  \cite{xue2006additive},   \cite{xue2009consistent},  \cite{wang2007group,wang2008variable} and \cite{huang2012selective} for nonconvex penalties such as group SCAD and MCP.  Recently, a new avenue of research, commonly referred to as the convex-nonconvex strategy, has been proposed to mitigate the drawbacks of nonconvex penalties. Especially, \cite{selesnick2017sparse} proposed a generalized minimax concave penalty for the regularized least squares problems, the properties of which are investigated further by \cite{liu2021convex}. However, the properties of additive models are still being investigated using this penalty.  Similarly, \cite{chouldechova2015generalized} proposed a traditional greedy component selection method. We also refer to \cite{amato2016additive} for a recent review of the model selection with additive models and \cite{sadhanala2019additive} for the most recent contributions to trend filtering.

In the context of high-dimensional regression, outliers cannot be effectively detected using scatter plots of the predictor variables. 
To identify outliers, it is necessary to fit a regression model and estimate the associated error terms. A key limitation of classical regression methods is that they tend to be influenced by outliers, which can obscure the detection of moderate outliers while misleadingly identifying normal observations as outliers. This masking effect is exacerbated in non-linear models. Consequently, a robust approach is essential for handling outliers effectively.

Most of the above estimation methods use the $L_2$ loss function, which is very sensitive to outliers in the data. While there has been a growing interest, the literature on robust estimation of additive models is limited. \cite{bianco1998robust} proposed robust kernel estimators for additive models.   \cite{alimadad2011outlier} considered a robust quasi-likelihood approach to estimate additive models in the presence of outliers. \cite{croux2012robust} propose a robust estimation approach using penalized splines. \cite{wong2014robust} consider a penalized M-estimation approach for additive models. \cite{boente2017marginal} consider a marginal integration method with robust local polynomial fitting, and \cite{boente2017robust} proposes robust estimators using a backfitting approach. Recently, \cite{boente2023robust} proposed a robust estimation procedure for the partial linear additive model using B-spline basis functions. Most of these approaches discuss the robust estimation of additive models in the finite-dimension case and do not discuss the model selection. In a recent study, \cite{amato2022wavelet} proposed a robust variable selection method using wavelets.

This study considers a new approach for robust estimation of a high-dimensional additive model using density power divergence (DPD) measure with a nonconcave penalty function. \cite{MR1665873} introduced the DPD measure as a robust extension of the well-known Kullback-Leibler divergence. The estimator obtained from the DPD measure is an M-estimator and produces improved performance in the presence of outliers and a heavy-tailed distribution. \cite{mandal2019robust,mandal2022robust} proposed a DPD-based robust variable selection method for multiple linear regression (MLR). 
However, to the best of our knowledge, an efficient method for the robust variable selection for the nonparametric additive model is missing in the statistics literature.

Our contributions to this study are twofold.
\begin{itemize}
\item The proposed method considers the DPD loss function and employs a nonconcave penalty. Given the computational simplicity of the DPD loss function, the proposed method is computationally faster. Moreover, it is easier to interpret the DPD estimator than other M-estimators, as the former is derived from a distance-like measure. 

    \item In terms of theory, our results assume the error to be sub-Weibull, which is a more general condition than the sub-Gaussian or the sub-Exponential conditions used in the literature. Our theoretical results are new to the literature on nonparametric additive models.
\end{itemize}

The remainder of the paper is organized as follows. In Section \ref{sec:dpd}, we provide a brief introduction to the density power divergence and the folded-concave penalty function. The estimation method using polynomial splines and the penalized likelihood is discussed in Section \ref{sec:method}. In Section \ref{sec:asymp}, we provide the asymptotic results of the proposed estimators. Section \ref{sec:simulations} presents the finite sample performance of the proposed method using simulations, and Section \ref{sec:data-analysis} illustrates the performance of the proposed method using a real example. The Supplementary Material includes the proofs and the necessary lemmas for the asymptotic results.

\section{Background} \label{sec:dpd}
This section briefly introduces the density power divergence measure \cite{MR1665873}, which estimates model parameters robustly, and the nonconcave penalty function \cite{lv2009unified}, which applies regularization.

\subsection{Density Power Divergence (DPD)}

A density-based power divergence is a family of statistical ``distance'' measures used to compare the goodness-of-fit of a model distribution to a reference distribution.
Suppose the true data-generating distribution is $m$,  and we model it with a parametric distribution $f_\theta$ for $\theta \in \Theta$. Then, \cite{MR1665873} defined the DPD measure between $f_\theta$ and $m$ as
	\begin{equation}
	d_{\nu}(f_\theta, m) = 
	\left\{
	\begin{array}{ll}
		\int_y\left\{ f_\theta^{1+\nu}(y)-\left( 1+\frac{1}{\nu}\right) f_\theta^{\nu }(y)m(y)+%
		\frac{1}{\nu}m^{1+\nu}(y)\right\} dy, &\text{for}\mathrm{~}\nu>0, \\%
		[2ex]
		\int_y m(y)\log\left( \displaystyle\frac{m(y)}{f_\theta(y)}\right) dy, & \text{for}%
		\mathrm{~}\nu=0,%
	\end{array}
	\right. 
	\label{eqn:dpd}
	\end{equation}    
	where $\nu$ is a tuning parameter. 	For $\nu=0$, the DPD is obtained as a limiting case of $\nu \rightarrow 0^+$, and the measure is called the  Kullback-Leibler divergence. 
	Given a parametric model $f_\theta$, we estimate $\theta$ by minimizing the DPD measure with respect to $\theta$ over its parametric space $\Theta$. We call the estimator $(\widehat{\theta})$ the minimum density power divergence estimator (MDPDE). 
	For $\nu=0$, minimizing the DPD loss function (\ref{eqn:dpd}) is equivalent to maximizing the log-likelihood function. Therefore, the MLE is a special case of the MDPDE. 

Consider a scenario where $m$ contains an $\epsilon$ proportion of outliers from an arbitrary distribution $\chi$, expressed as $m(y) = (1-\epsilon) f_{\theta_0}(y) + \epsilon \chi(y)$, with $0 \leq \epsilon < 0.5$ and $\theta_0 \in \Theta$ denoting the true value of the parameter. We assume the existence of a small positive constant $\nu_0$, such that $\eta(\nu) = \int_{y} f_{\theta_0}^\nu(y) \chi(y) \, dy$ is sufficiently small for $\nu > \nu_0$. This small value of $\eta(\nu)$ confirms that $\chi$ is an outlying distribution, as its effective mass resides in the tail of the model distribution $f_{\theta_0}$. Under this framework, the MDPDE with $\nu > \nu_0$ effectively mitigates the influence of outliers, resulting in a robust and consistent estimator \cite{fujisawa2008robust,das2022testing}. Additionally, the tuning parameter $\nu$ controls the trade-off between efficiency and robustness of the MDPDE; as $\nu$ increases, the robustness measure improves, but this is accompanied by a reduction in efficiency. A data-dependent approach for selecting the optimal $\nu$ is explored in Section \ref{sec:opt_alpha}.

\subsection{Nonconcave Penalty}
There are multiple penalty functions available in the literature for model regularization. Popular approaches include the $L_2$ penalty, which provides the ridge regression estimates, and the $L_0$ penalty, which yields the best subset selection. Then, there is $L_q$, $q \in (0,2)$, which bridges the gap (bridge regression) between them \cite{frank1993statistical}. Similarly, LASSO uses the $L_1$ penalty \cite{tibshirani1996regression}. On the other hand, \cite{lv2009unified} introduce a family of folded-concave penalties, including the SCAD penalty \cite{fan1997comments,fan2001variable}, that bridges $L_0$ and $L_1$ penalties.
In a seminal paper, \cite{fan2001variable} discuss that penalty functions yield estimators with the following three properties:
\begin{itemize}
    \item  \emph{Unbiasedness:} When the absolute value of the true parameter is large in magnitude, the penalty function produces nearly unbiased estimators.
    \item \emph{Sparsity:} The penalty function automatically sets the small estimated coefficients to zero to reduce the model dimension.
    \item \emph{Continuity:} The resulting estimator should be continuous in data to avoid instabilities in model predictions. 
\end{itemize}
Let $P_\lambda(t)$ denote the penalty function with the penalty parameter $\lambda$. The derivative of the SCAD penalty is 
\begin{align*}
    P_{\text{SCAD}, \lambda}'(t) &= \lambda \left\{ I(t \le \lambda) + \frac{(a\lambda -t)_+}{(a-1)\lambda} I(t > \lambda) \right\}, \quad t \ge 0,
\end{align*}
for some $a>2$, where often $a=3.7$ is considered in practice. The derivative of the minimax concave penalty (MCP) \cite{zhang2010nearly} is given as $P_{\text{MCP},\lambda}'(t)=(a\lambda-t)_+/a$. It is noted that the MCP translates the flat part of the derivative of SCAD to the origin. On the other hand, the SCAD penalty takes off at the origin as the $L_1$ penalty and then levels off \cite{fan2011nonconcave}. In this study, we consider penalty functions $P(t;\lambda)=\lambda^{-1}P_{\lambda}(t)$ that satisfy the following condition, which is considered in \cite{lv2009unified}.

\emph{Condition 1:} The penalty function $P(t;\lambda)=\lambda^{-1}P_{\lambda}(t)$ is increasing and concave in $t \in [0,\infty)$, and has a continuous derivative $P'(t;\lambda)$ with $P'(0+, \lambda)>0$. Furthermore, $P'(t;\lambda)$ increases in $\lambda \in (0, \infty)$ and $P'(0+; \lambda)$ is independent of $\lambda$.

We further note that, as mentioned in \cite{fan2011nonconcave}, while SCAD, MCP, and $L_1$ penalties satisfy Condition 1, the $L_1$ penalty is not unbiased, and MCP violates the continuity property.

\section{Methodology}\label{sec:method}

\subsection{Spline Basis}
Without loss of generality, assume that each covariate $X_j$ in (\ref{eqn:model-def}) takes values in $[0,1]$. Similar to \cite{huang2010variable}, we use polynomial splines to approximate the component functions $g_j(\cdot)$, $j=1,2,\ldots,p$, in (\ref{eqn:model-def}). Specifically, we consider the B-spline basis due to its computational simplicity; however, the proposed approach can be extended to other basis functions with suitable modifications. Let $\mathcal{G}_n$ be the space of polynomial splines on $[0,1]$ of order $r \ge 1$. Let 
\begin{align*}
 t_{-r+1}=&\cdots=t_{-1}=t_0 = 0 < t_1 < \cdots < t_K < 1=t_{K+1}= \cdots = t_{K+r},   
\end{align*}
be a knot sequence with $K$ interior knots, where $K \equiv K_n$ increases with sample size $n$, such that $\max_{1 \le l \le K+1 }| t_l-t_{l-1}|=O(1/K_n)$.  As defined in \cite{stone1985additive} and \cite{schumaker2007spline}, $\mathcal{G}_n$ consists of functions $g(\cdot)$ that satisfy the following properties:
\begin{enumerate}
    \item the restriction of $g(\cdot)$ in each interval $I_{l}=[t_l,t_{l+1}]$, $l=0, 1,\ldots,K_n$,  is a polynomial of degree $r-1$; 
    \item $g(\cdot)$ is $r-2$ times continuously differentiable on $[0,1]$ for $r \ge 2$.
\end{enumerate}
Then, there exists a normalized B-spline basis $\{b_{k}, 1 \le k \le K_n+r\}$  in $\mathcal{G}_n$ such that for any function $g_{n0j} \in \mathcal{G}_n$, we can write 
\begin{align}
    g_{n0j}(x) &= \sum_{k=1}^{m_n}\beta_{kj}b_{k}(x), \ 1 \le j \le p, \text{  and  } m_n=K_n+r-1.
    \label{eqn:mn}
\end{align}
For simplicity in proofs, equally-spaced knots are considered; however, other knot sequences may be used with similar asymptotic results.
Similar to \cite{li2012scad}, to ensure the identifiability of the component functions, we adopt the normalized B-spline space, denoted by $\mathcal{G}_n^0$  
\begin{align*}
    \left\{ B_{kj}(x) =  b_{k}(x)-  \frac{1}{n}\sum_{i=1}^n b_{k}(X_{ij}), ~  1 \le k \le m_n, 1 \le j \le p \right\}.
\end{align*}

For convenience in notation, let
\begin{equation*}
    Z_i =(1, B_{11}(X_{i1}), \ldots, B_{m_n1}(X_{i1}), \ldots, B_{1p}(X_{ip}), \ldots, B_{m_np}(X_{ip}))^T,
    \end{equation*}
\vspace{-.5cm}
    \begin{equation*}
    \beta_{j} =(\beta_{1j}, \ldots, \beta_{m_nj})^T,\quad B_j=(B_{1j}, \ldots, B_{m_nj})^T,
\end{equation*}
and $Z=(Z_1, \ldots,Z_n)^T$, $\beta^T=(\mu, \beta_1^T, \ldots, \beta_p^T)$. Then, we approximate the additive function $g_j(\cdot)$ in model \eqref{eqn:model-def} as
\begin{align}
    g_{j}(x) &\approx  g_{nj}(x) := \sum_{k=1}^{m_n} \beta_{kj}B_{kj}(x)= B_j^T\beta_j, \mbox{ for }  j = 1, 2, \ldots, p, \label{eqn:norm-splinebasis}
\end{align}
where the basis functions $B_{kj}(x) \in \mathcal{G}_n^0$.

\subsection{Penalized Likelihood using DPD Framework}
As we are dealing with a robust estimator, we do not assume that the true distribution $m(\cdot)$ in (\ref{eqn:dpd}) or the error term $\epsilon_i$ in model \eqref{eqn:model-def}  is a normal distribution, but it is a distribution with an exponentially heavy-tail (please see assumption \ref{as:5-error} in Section \ref{sec:asymp}). However, for simplicity, the target distribution $f_\theta(\cdot)$ in (\ref{eqn:dpd}) is assumed to be $N(0,\sigma^2)$ for all $i=1, 2, \cdots, n$. The theory proceeds similarly to other choices of the target distribution $f_\theta(\cdot)$ under mild conditions. Thus, the model parameter $\theta = (\beta^T, \sigma^2)^T$, with $\theta \in \Theta$, is robustly estimated to match the target distribution. 
We denote the target distribution of $Y_i$, i.e., $N(Z_i^T \beta, \sigma^2)$, as $f_{\theta}(Y_i|Z_i)$, or simply
\begin{align}
    f_i := f(\beta; X_i, Y_i)= \frac{1}{\sqrt{2\pi}\sigma} \exp\left(-\frac{1}{2 \sigma^2}\left(Y_i-Z_i^T\beta \right)^2\right), ~ i=1,2,\ldots,n. \label{eqn:fi}
\end{align}



Using Equation  (\ref{eqn:norm-splinebasis}), the minimization problem for model (\ref{eqn:model-def}) with the DPD loss function (\ref{eqn:dpd}), for $\nu>0$, can be formulated as 
\begin{align}
    L_{\nu}( \beta, \sigma^2; \lambda) &= \frac{1}{n} \sum_{i=1}^n V_i(\beta,\sigma^2; \lambda, \nu) + \mathcal{P}_{\lambda} + c(\nu), \label{eqn:plike}    
\end{align}
where $c(\nu)= \nu^{-1}\int_y m^{1+\nu}(y)dy$, which is free of the model parameters, and
\begin{align}
    V_i(\beta,\sigma^2; \lambda, \nu) &:= \frac{1}{(2\pi)^{\nu/2} \sigma^{\nu}\sqrt{1+\nu}}- \frac{1+\nu}{\nu}f_i^{\nu}.
    \label{eqn:vi}
\end{align}
Here, the penalty term in (\ref{eqn:plike}) is defined as
\begin{align*}
    \mathcal{P}_{\lambda} &= \sum_{j=1}^p P_{\lambda}(\Vert \beta_{j} \Vert_2),
\end{align*}
where  $\Vert \beta_{j} \Vert_2^2= \beta_j^T D_j \beta_j$ with $(k,k')$th entry of $D_j= \int_0^1  B_k(t) B_{k'}^T(t) dt$  and $P_{\lambda}(\cdot)$ is a nonconcave penalty function defined in Condition 1 in Section \ref{sec:dpd}. The estimated component functions can be expressed as $\widehat{g}_{nj}= B_j^T\widehat{\beta}_j$ for $j=1,2,\ldots,p$.

The estimators of $\beta$ and $\sigma^2$ in (\ref{eqn:plike}) are obtained by minimizing $L_{\nu}(\beta, \sigma^2; \lambda)$ over $ \beta \in \mathbb{R}^{pm_n+1}$ and $\sigma >0$. If the $i$-th observation is an outlier, the value of $f_i$ is very small compared to other samples. When $\nu > 0$, the second term of \eqref{eqn:vi} is negligible for that $i$; thus, the resulting estimator becomes insensitive to the outlier. On the other hand, when $\nu=0$, we have $V_i(\beta,\sigma^2; \lambda, \nu) = -\log(f_i)$, and it diverges as $f_i \rightarrow 0$. Hence, the maximum likelihood breaks down in the presence of outliers because they dominate the loss function.

\subsection{DPD Estimation Algorithm} \label{sub_sec_dpd_algo}
Differentiating the expression \eqref{eqn:plike} with respect to $\beta$ and $\sigma^2$, we obtain the following normal equations:
	\begin{align} 
		-   \sum_{i=1}^n (Y_i - Z_i^T \beta ) Z_i f_i^\nu
		+  \partial \mathcal{P}_{\lambda} &= 0,
		\label{est_beta}\\
		-\frac{n\nu}{(2\pi)^{\nu/2} \sigma^\nu (1+\nu)^{3/2}} +   \sum_{i=1}^n &\left\{1 - \frac{(Y_i - Z_i^T \beta )^2}{\sigma^2} \right\} f_i^\nu = 0,
		\label{est_sigma}
	\end{align}
 where $\partial \mathcal{P}_{\lambda}$ is the subdifferential of $\mathcal{P}_{\lambda}$ with respect to $\beta$. 
We define a diagonal matrix $W_\nu$ with diagonal elements $f_1^\nu, f_2^\nu, \cdots, f_n^\nu$. When $\lambda=0$, rearranging terms in  \eqref{est_beta} and \eqref{est_sigma}, we find the unpenalized minimum density power divergence estimator (unpenalized MDPDE), $(\widetilde{\beta}, \widetilde{\sigma}^2)$, using the following iterative algorithm:
	\begin{align} 
\widetilde{\beta} &=  (Z^T \widehat{W}_\nu Z)^{-1} Z^T \widehat{W}_\nu Y,\\
\widetilde{\sigma}^2 &= \frac{ 1 }{ \text{trace}(\widehat{W}_\nu) } \Bigg[\frac{1}{n} (Y - Z^T \widetilde{\beta} )^T \widehat{W}_\nu (Y - Z^T \widetilde{\beta} ) + \frac{\nu}{(2\pi \widehat{\sigma}^2)^{\nu/2}  (1+\nu)^{3/2}}\Bigg], \label{mdpde_sigma}
	\end{align}
 where $Y=(Y_1,\ldots,Y_n)^T$.
Now, similar to the SCAD estimator in the case of the square-error loss function, we derive $\widehat{\beta}_{SCAD}$, the penalized MDPDE of $\beta$ using the SCAD penalty function. Let us define $Y^* = Y - \widehat{\mu}$, where $\widehat{\mu}$ is the intercept of unpenalized MDPDE. Setting $Z_j^T \widehat{W}_\nu Z_j = I_{m_{nj}}$ for all $j=1,2,\cdots, p$, the penalized MDPDE  $\widehat{\beta}_{j, SCAD}$ corresponding  to the $j$-th predictor is obtained as
\begin{equation}
    \widehat{\beta}_{j, SCAD}  =  F_{\lambda, \gamma}^{SCAD}(Z_j^*),
\end{equation} 
where $Z_j^* = Z_j^T W_\nu^{1/2} (Y^* - Z_{-j} W_\nu^{1/2} \beta_{-j})$ with $Z_{-j}$ denoting all columns of $Z$ excluding $Z_j$; similarly, $\beta_{-j}$ denotes all elements of $\beta$ excluding $\beta_j$, and for $\gamma > 2$,
\begin{equation}
F_{\lambda, \gamma}^{SCAD}(x) =  \left\{
	\begin{array}{ll}
S(x,\lambda), & \text{ if } \Vert x \Vert \leq 2 \lambda,\\
S\left(x,\frac{\gamma \lambda}{\gamma -1}\right) \frac{\gamma - 1}{\gamma -2}, & \text{ if } 2\lambda < \Vert x \Vert \leq \gamma \lambda,\\
x, & \text{ if } \Vert x \Vert > \gamma \lambda,
\end{array}
	\right.     
\end{equation}    
where
\begin{equation}
S(x, t) =  \left( 1 -\frac{t}{\Vert x \Vert} \right)_{+} x,     
\end{equation} 
is a multivariate soft-threshold operator \cite{huang2012selective}. Finally, $\widehat{\sigma}^2$, the penalized MDPDE of $\sigma^2$, is estimated from Equation \eqref{mdpde_sigma} replacing $\widetilde{\beta}$ by $\widehat{\beta}_{SCAD}$. The expressions for the estimator using group MCP and group LASSO penalties are also obtained in a similar way.

\subsection{ Model Selection Criteria}

To implement our algorithm, we start with a fixed number of spline basis functions, $m_n$, for all covariates. We need an adaptive procedure to choose the regularization parameter $\lambda$ based on a given data set. While cross-validation (CV) is used for this purpose, it is not appropriate in robust regression since the test data also contain outliers. Therefore, we apply different information criteria to find the optimum $\lambda$. Specifically, we consider four methods: Akaike information criterion (AIC), Bayesian information criterion (BIC), extended Bayesian information criterion (EBIC), and Mallow's Cp statistic \cite{cl1973some,li2012scad}, which are defined as follows:   
\begin{align} \label{IC}
    \text{AIC}(\lambda) &= n \log (n \widehat{\sigma}^2) + 2 \ \text{df},\nonumber\\
    \text{BIC}(\lambda) &= n \log (n \widehat{\sigma}^2) +   \text{df}\ \log n,\nonumber\\
    \text{EBIC}(\lambda) &= n \log (n \widehat{\sigma}^2) +  (\log n + \log p) \ \text{df},\\
    \text{Cp}(\lambda) &= \frac{n\widehat{\sigma}^2}{\widehat{\sigma}_u^2} - n + 2 \ \text{df},\nonumber
\end{align}
where $\widehat{\sigma}^2$ is the estimate of $\sigma^2$ obtained from the sub-model using the algorithm given in Section \ref{sub_sec_dpd_algo}, and $\widehat{\sigma}_u$ is an unbiased and robust estimator of $\sigma$ preferably from the full model. Here, df is the degrees of freedom of the sub-model, i.e. the dimension of non-zero $\beta$ coefficients obtained from the group SCAD estimator. It should be noted that $\widehat{\sigma}^2$ and df are functions of $\lambda$. We select the optimum $\lambda$ for each case that minimizes the corresponding information criterion.

The AIC and Cp have similar performance as they are equivalent in the linear models. The BIC imposes more penalty on the model than the AIC when $n>7$. We prefer EBIC in high-dimensional data, as it adjusts the bias due to increased dimensionality \cite{chen2008extended}. As our method is not sensitive to the choice of the information criteria, one may use other criteria, e.g.,  GIC \cite{konishi1996generalised}, ERIC \cite{hui2015tuning} and GBIC \cite{konishi2004bayesian}. We may also consider a robust information criteria in this setup \cite{mandal2022robust}. 


\subsection{Choosing the DPD parameter}\label{sec:opt_alpha}
The performance of the MDPDE depends on the choice of the tuning parameter $\nu$ that controls the robustness and efficiency of the estimator. If the data contain no outliers, the MLE or a small value of $\nu$ produces the best results. On the other hand, in the presence of outliers, we need a large value of $\nu$. The optimum value of $\nu$ depends on the proportion of outliers and their magnitude. The main challenge is that we rarely have any prior knowledge of outliers. Therefore, we need an adaptive choice of $\nu$ that is optimum for the given data set. In the parametric model, one can estimate the mean square error (MSE) of the parameter of interest and minimize it iteratively as a function of $\nu$ \cite{mandal2023robust}. However, finding the MSE of the additive components in model \eqref{eqn:model-def} is difficult. Therefore, we will use the Hyvarinen score proposed by \cite{hyvarinen2005estimation}. It is easy to calculate, as it only requires the first and second-order partial derivatives of the DPD measure for an observation.  

The Hyvarinen score constructs a pseudo-likelihood based on the DPD measure. Note that the empirical version of DPD in Equation \eqref{eqn:plike} can be written as a function of $y_i$ as $\sum_{i=1}^n V^*_\nu(y_i; \beta, \sigma^2, \lambda)$, where $n V^*_\nu(y_i; \beta, \sigma^2, \lambda) = V_i(\beta,\sigma^2; \lambda, \nu) + \mathcal{P}_{\lambda} + c(\nu)$. So, we assume an unnormalized density of $y_i$ as $\exp[ V^*_\nu(y_i; \beta, \sigma^2, \lambda)]$. Then, following \cite{sugasawa2021selection}, the Hyvarinen score is given as
\begin{small}
\begin{align}
    S_n(\nu) &= \frac{1}{n} \sum_{i=1}^n \left\{ - 2 \frac{\partial^2}{\partial y_i^2} V^*_\nu(y_i; \beta, \sigma^2, \lambda) + \left(\frac{\partial}{\partial y_i} V^*_\nu(y_i; \beta, \sigma^2, \lambda) \right)^2 \right\} \nonumber\\
&= \frac{1+\nu}{n \sigma^4} \sum_{i=1}^n \left\{\nu(y_i - x_i^T \beta)^2 - \sigma^2+ (1+\nu) (y_i - x_i^T \beta)^2 f_i^\nu \right\} f_i^\nu .
\label{h_score}
\end{align}    
\end{small}
The optimum $\nu$ is selected that minimizes $S_n(\nu)$. We use a grid search for this purpose. For a given data set, we first find the optimum $\lambda$ that minimizes an information criterion from \eqref{IC}. Then, the Hyvarinen score is calculated from Equation \eqref{h_score} by plugging in the estimates of $\beta$ and $\sigma$. Our final DPD estimator is the one that minimizes the Hyvarinen score.

\section{Asymptotic results} \label{sec:asymp}

In this section, we present the asymptotic results for the estimators defined in Section \ref{sec:method}. We need the following conditions for our asymptotic results. 
\begin{assumption} \label{as:1-function-class}
The additive functions $\{g_j\}_1^p$ belong to a class of functions $\mathcal{H}$ whose $r$-th derivative exists and is Lipschitz of order $u$:
\begin{align*}
    \mathcal{H} &= \left\{ g(\cdot):~ |g^{(r)}(s)- g^{(r)}(t)| \le C|s-t|^{u}, \text{  for } s, t \in [0,1] \right\}
\end{align*}
for some positive constant $C$, where $r$ is a non-negative integer and $u \in (0,1]$ such that $d=r+u >2$.
\end{assumption}

\begin{assumption} \label{as:2-density-x}
The distribution of $X$ is absolutely continuous with its density $h(\cdot)$ is bounded away from zero and infinity on $[0,1]^p$. 
\end{assumption}

\begin{assumption} \label{as:3-ncomp}
The number of nonzero components $q$ is fixed. 
\end{assumption}


\begin{assumption} \label{as:5-error}
The random variables $\epsilon_1, \epsilon_2, \ldots, \epsilon_n$ are independent and identically distributed with mean zero and finite variances. Moreover $\Vert \epsilon_i \Vert_{\psi_{\alpha}} < \infty$, $i=1,\ldots,n$, for  $\alpha >0$, where
\begin{align}
    \Vert \epsilon_i \Vert_{\psi_{\alpha}} &:= \inf\left\{ \eta>0~:~ \mathbb{E}[\psi_{\alpha}(|\epsilon_i|/\eta)] \le 1 \right\}, \label{eqn:orlicz}
\end{align}
with $\psi_{\alpha}(x)= \exp(x^{\alpha})-1$, for $x \ge 0$.
\end{assumption}

\begin{assumption} \label{as:6-penalty}

The maximum concavity of the penalty function $P(t;\lambda)$ (Condition 1 in Section \ref{sec:dpd}) defined as 
\begin{align}
    \kappa_0 &=  \underset{t_1 < t_2 \in (0, \infty )}{\sup} - \frac{P_{\lambda}'(t_2)-P_{\lambda}'(t_1)}{\lambda(t_2-t_1)}, \label{eqn:kappa}
\end{align}
 satisfies $\lambda \kappa_0 = o(m_n)$.

\end{assumption}

Assumptions \ref{as:1-function-class} and \ref{as:2-density-x} are standard in the additive modeling literature \cite{stone1985additive,xue2009consistent,huang2010variable,fan2011nonparametric,wang2011estimation,zhong2020forward}. While assumption \ref{as:1-function-class} is a smoothness condition for the component function, which is required for computing the spline approximation error, assumption \ref{as:2-density-x} makes sure that the  marginal densities of the random variables $X_j$ are bounded away from zero and infinity.  Assumption \ref{as:3-ncomp} is a sparsity condition, which is common in the high-dimensional regression modeling literature. For example, please see \cite{huang2010variable}. In assumption \ref{as:5-error}, we assume the error is sub-Weibull. This assumption is weaker than the sub-Gaussian assumption, which is commonly used in the literature \cite{huang2010variable,kuchibhotla2022moving}. We define the sub-Weibull random variable using the Orlicz norm in (\ref{eqn:orlicz}). Both sub-Gaussian and sub-Exponential are special cases of (\ref{eqn:orlicz}) for $\alpha=2$ and $\alpha=1$, respectively. We refer to Section 1 of  \cite{wellner2017bennett} for a historical account of the Orlicz norm and sub-Weibull random variables. Finally, assumption \ref{as:6-penalty} is a mild condition needed for the folded-concave penalty function \cite{fan2011nonconcave} for which the second derivative may not exist. For $L_1$, SCAD and MCP  penalties,  $\lambda \kappa_0$ takes values $0$, $(a-1)^{-1}$, and $a^{-1}$, respectively.

Similar to \cite{fan2011nonparametric}, under assumptions \ref{as:1-function-class} and \ref{as:2-density-x}, we state the following facts that are needed for our asymptotic results.
\begin{itemize}
    \item[] Fact 1. \cite{stone1985additive} There exists  positive constants $C_1$ and $C_{11}$ such that
    \begin{align}
        \mathbb{E} \Vert g_j-g_{nj} \Vert_2^2 \le C_1 m_n^{-2d} \text{ and } \Vert g_j-g_{nj}  \Vert_{\infty} \le C_{11} m_n^{-d}, \label{eqn:fact1}
    \end{align}
    where $m_n$ is defined in (\ref{eqn:mn}).
    \item[] Fact 2. (Theorem 5.4.2 of \cite{devore1993constructive}) For $s > 1$, there exists a positive constant $C_{s}$ (depends on $s$) such that
    \begin{align}
        \mathbb{E} B_{kj}^{s}(X_{ij}) \le C_{s} m_n^{-1}. \label{eqn:fact2}
    \end{align}
\end{itemize}


We call a model ``oracle" if the true nonzero components are known, and denote their estimators with superscript ``0" notation. In the following theorem, we provide the results for the oracle model. Without loss of generality, we assume that the first $q$ additive components $\{g_j\}_1^q$ are nonzero.  We now present the convergence rate of the oracle estimator, $\widehat{\beta}_j^0$, where $\widehat{g}^0_j=B_j^T\widehat{\beta}_j^0$, $j=1,2,\ldots,q$.

\begin{theorem}\label{thm:oracle}
Suppose the assumptions \ref{as:1-function-class} --\ref{as:5-error} hold.  If $m_n \log(qm_n)/n \rightarrow 0$, $m_n \rightarrow \infty$, as $n \rightarrow \infty$ then, for $\alpha>0$, we have
\begin{align}
    \sum_{j=1}^q \Vert \widehat{\beta}_{j}^0 - \beta_{j} \Vert^2 &= O_p\left( \frac{m_n^2 \log(q m_n) }{n}\right) + O_p\left( \frac{m_n^{3-2\alpha^*} (\log(q m_n))^{2/\alpha} }{n^{2-2\alpha^*}}\right) \nonumber \\ & \qquad  + O_p\left( \frac{m_n}{n} \right) + O\left( \frac{1}{m_n^{2d-1}} \right), \label{eqn:oracle-final}
\end{align}    
where $\alpha^*=\max\{0, \frac{\alpha-1}{\alpha}\}$.
\end{theorem}

The rate of convergence in Theorem \ref{thm:oracle} is determined by four terms: the first two terms contribute to the stochastic error in estimating the nonparametric components, the third term is for estimating intercept $\mu$, and the fourth term is for the spline approximation error. When the error ($\epsilon$) is sub-Gaussian $(\alpha=2, \alpha^*=1/2)$, the second term in (\ref{eqn:oracle-final}) coincides with the first term, and it is the same convergence rate obtained in Theorem \ref{thm:oracle} of \cite{huang2010variable}, where the authors' used sub-Gaussian assumption.

In the following theorem, we provide the convergence rate of the penalized estimator $\widehat{\beta}_j$, which yields the estimated additive function $\widehat{g}_j=B_j^T\widehat{\beta}_j$, $j=1,2,\ldots,p$.
\begin{theorem} \label{thm:consistency}
Suppose the assumptions \ref{as:1-function-class} --\ref{as:6-penalty} hold. If $m_n \log(pm_n)/n \rightarrow 0$, $m_n \rightarrow \infty$ and $\lambda^2m_n \rightarrow 0$ as $n \rightarrow \infty$, then there exists a strict local minimizer $\widehat{\beta}=(\widehat{\beta}^{(1)^T}, \widehat{\beta}^{(2)^T} )^T$, with $\widehat{\beta}^{(1)}=(\widehat{\mu},\widehat{\beta}_1^T, \ldots ,\widehat{\beta}_q^T)^T$ and $\widehat{\beta}^{(2)} \in \mathbb{R}^{(p-q)m_n}$, of the penalized DPD loss $L_{\alpha}$ in (\ref{eqn:plike}) such that $\widehat{\beta}^{(2)}=0$ with probability tending to 1 as $n \rightarrow \infty$. Further, for $\alpha>0$, we have
\begin{align}
    \sum_{j=1}^q \Vert \widehat{\beta}_{j} - \beta_{j} \Vert^2 &= O_p\left( \frac{m_n^2 \log(p m_n)}{n}\right)  + O_p\left( \frac{m_n^{3-2\alpha^*} (\log(p m_n))^{2/\alpha} }{n^{2-2\alpha^*}}\right)  \nonumber \\ & \qquad + O_p\left( \frac{m_n}{n} \right) + O\left( \frac{1}{m_n^{2d-1}}  \right)+ O(q \lambda^2), \label{eqn:consistent-final}
\end{align}
where $\alpha^*=\max\{0, \frac{\alpha-1}{\alpha}\}$.
\end{theorem}
The rate of convergence in (\ref{eqn:consistent-final}) has an additional term due to penalization than the rate in (\ref{eqn:oracle-final}). It is worth mentioning that the rate for the penalization term $q \lambda^2$ is different from the rate for the same $4 m_n^2\lambda^2/n^2$ in Theorem \ref{thm:oracle} of \cite{huang2010variable}. It is due to the normalization constant ($n$) in the penalized likelihood expression (\ref{eqn:plike}), which is not present in the penalized objective function of \cite{huang2010variable}. Theorem \ref{thm:consistency} only proves the existence of a local minimizer of the objective function (\ref{eqn:plike}), which is typical in the literature with nonconcave penalties, for example, Theorem \ref{thm:consistency} of \cite{xue2009consistent}.

The following corollary is a direct consequence of Theorem \ref{thm:consistency} using the properties of B-splines. Here, the estimated additive function $\widehat{g}_j=B_j^T\widehat{\beta}_j$, $j=1,\ldots,q$.
\begin{corollary}
Suppose the conditions in Theorem \ref{thm:consistency} hold. Then, for $\alpha>0$, we have
\begin{align*}
    \sum_{j=1}^q \Vert \widehat{g}_{j} - g_{j} \Vert^2 &= O_p\left( \frac{m_n \log(p m_n)}{n}\right)   + O_p\left( \frac{m_n^{2-2\alpha^*} (\log(p m_n))^{2/\alpha} }{n^{2-2\alpha^*}}\right) \\ & \qquad + O_p\left( \frac{1}{n} \right) + O\left( \frac{1}{m_n^{2d}}  \right)+ O(q \lambda^2m_n^{-1}),
\end{align*}
where $\alpha^*=\max\{0, \frac{\alpha-1}{\alpha}\}$.
\end{corollary}
Let's introduce an additional notation. For any two sequences $\{a_n, b_n, n=1,2,\ldots\}$, we denote $a_n \asymp b_n$ if there exists come constants $0< c_1 < c_2 < \infty$ such that $c_1 \le a_n/b_n \le c_2$ for all sufficiently large $n$. We now state a corollary useful for providing an upper bound on the number of zero component functions, which depends on the sub-Weibull parameter $\alpha$.
\begin{corollary} \label{cor:pbound}
    Suppose the conditions in Theorem \ref{thm:consistency} hold. In addition, assume that $m_n \asymp n^{1/(2d+1)}$ and $\lambda \asymp m_n^{3/2-\alpha^*} (\log(pm_n))^{1/\alpha}/n^{1-\alpha^*}$  for $\alpha > 0$ and $\alpha^*=\max\{0, (\alpha-1)/\alpha\}$. Then we have
    \begin{align*}
    \sum_{j=1}^q \Vert \widehat{g}_{j} - g_{j} \Vert^2 &= O_p\left( n^{-2d/(2d+1)} \log(p m_n)\right)  \\ & \qquad \qquad + O_p\left( n^{-4d(1-\alpha^*)/(2d+1)} (\log(p m_n))^{2/\alpha} \right).
\end{align*}
\end{corollary}
From Corollary \ref{cor:pbound}, we can identify that the number of zero components can be as large as  $\exp\left( o(n^{2d\alpha(1-\alpha^*)/(2d+1)}) \right)$. For sub-Gaussian error $(\alpha=2)$ this equals to $\exp(o(n^{2d/(2d+1)}))$. For example, if we assume that each function $g_j(\cdot)$ has a continuous second derivative $(d=2)$, the number of zero components can be as large as $\exp(o(n^{4\alpha(1-\alpha^*)/5}))$. It implies the number can be as large as $\exp(o(n^{4/5}))$ for $\alpha \ge 1$ and as large as $\exp(o(n^{4\alpha/5}))$ for $\alpha < 1$, which is of a smaller order. We need to pay this price for choosing an exponentially heavy-tailed error distribution.

\subsection{Influence Function}

Let $(X, Y) \sim H_0$ where $H_0$ is a model distribution. Let $\beta(\cdot)$ be a functional  defined as
\begin{align*}
    \beta(H_0) &= \underset{\beta,\sigma}{\arg \min}~\mathbb{E}_{H_0}[V(\beta, \sigma^2; \lambda, \nu)] + \mathcal{P}_{\lambda}, 
\end{align*}
where $V$ and $\mathcal{P}_{\lambda}$ are defined similar to  $V_i$ and $\mathcal{P}_{\lambda}$ in (\ref{eqn:plike}) for the random variables $(X,Y)$. It is a population version of the  loss function defined in (\ref{eqn:plike}). Introduced by Hampel \cite{hampel1974influence}, the influence function measures the effect of small contamination on the functional \cite{ollerer2015influence}. Let $\delta_{(X_0,Y_0)}$ denotes the probability mass function that assigns mass 1 to the observation $(X_0,Y_0)$. Let  $H_{\varepsilon}=(1-\varepsilon)H_0 + \varepsilon \delta_{(X_0, Y_0)}$ be a contaminated distribution for  $0 \le \varepsilon \le 1/2$.  The influence function of the functional $\beta(\cdot)$ is defined as
\begin{align*}
    \text{IF}((X_0,Y_0); \beta, H_0) &= \frac{\partial }{\partial \varepsilon} \left[ \beta\left((1-\varepsilon)H_0 + \varepsilon \delta_{(X_0, Y_0)} \right)\right]\vert_{\varepsilon=0},
\end{align*}
where $(X_0,Y_0)$ is the pointwise contamination. We proceed similar to \cite{ollerer2015influence} and derive the influence function for a fixed value of $\sigma$.  Consider the first order condition of the functional at the contaminated model $H_{\varepsilon}$, 
\begin{align*}
    0 &= -\frac{1+\nu}{\sigma^2} \mathbb{E}_{H_{\varepsilon}} \left[ f^{\nu}(\beta; X, Y) (Y-Z^T\beta(H_{\varepsilon}))Z \right] + \partial\mathcal{P}_{\lambda}(\beta(H_{\varepsilon})),
\end{align*}
where $\partial\mathcal{P}_{\lambda}(\beta(H_{\varepsilon}))$ is a vector of subdifferential of $\mathcal{P}_{\lambda}(\beta(H_{\varepsilon}))$ with respect to $\beta(H_{\varepsilon})$.
It follows that
\begin{align*}
    - & \frac{1+\nu}{\sigma^2} (1-\varepsilon) \mathbb{E}_{H_{\varepsilon}} \left[ f^{\nu}(\beta(H_{\varepsilon}); X,Y) (Y-Z^T\beta(H_{\varepsilon}))Z \right] \nonumber \\  \qquad &-   \frac{1+\nu}{\sigma^2}\varepsilon  f^{\nu}(\beta(H_{\varepsilon}); X_0,Y_0) (Y_0-Z_0^T\beta(H_{\varepsilon}))Z_0 + \partial\mathcal{P}_{\lambda}(\beta(H_{\varepsilon}))=0.  
\end{align*}
Differentiating the above with respect to $\varepsilon$, we obtain
\begin{align*}
     \text{IF}((X_0,Y_0); \beta, H_0) &= \frac{1+\nu}{\sigma^2}S^{-1} \bigg[ \mathbb{E}_{H_0}\left[ f^{\nu}(\beta(H_{0}); X,Y) (Y-Z^T\beta(H_{0}))Z \right] \\ & \qquad - f^{\nu}(\beta(H_{0}); X_0,Y_0) (Y_0-Z_0^T\beta(H_{0}))Z_0  \bigg],
\end{align*}
where 
\begin{align*}
    S &= \frac{1+\nu}{\sigma^2} \bigg[ \mathbb{E}_{H_0}[f^{\nu}(\beta(H_{0}); X,Y) (Y-Z^T\beta(H_{0})) ~ (-ZZ^T) ] \\ &\qquad + \mathbb{E}_{H_0}\left[ \frac{\nu}{\sigma^2} f^{\nu}(\beta(H_{0}); X,Y) (Y-Z^T\beta(H_{0}))^2 ZZ^T \right] \bigg]- \partial^2\mathcal{P}_{\lambda}(\beta(H_{0})),
\end{align*}
with,
\begin{align*}
   & \partial^2\mathcal{P}_{\lambda}(\beta(H_{0})) = \text{diag}\bigg\{ \frac{P'_{\lambda}(\Vert \beta_j(H_{0}) \Vert_2)}{\Vert \beta_j(H_{0}) \Vert_2} D_j \\ & \qquad + \left[ \frac{P^{''}_{\lambda}(\Vert \beta_j(H_{0}) \Vert_2)}{\Vert \beta_j(H_{0}) \Vert_2^2} -\frac{P'_{\lambda}(\Vert \beta_j(H_{0}) \Vert_2)}{\Vert \beta_j(H_{0}) \Vert_2^3} \right] D_j \beta_j(H_{0}) \beta_j^{T}(H_{0}) D_j \bigg\},
\end{align*}
for $\Vert \beta_j(H_{0}) \Vert_2 \neq 0$. It can be shown that the expectations in $S$ and the influence functions are bounded under mild conditions. Then, the influence function is bounded in both $X_0$ and $Y_0$ since the function $xe^{-x^2}$ is bounded for $x \in \mathbb{R}$. Therefore, the DPD estimator $\hat{ \beta}$ is robust to outliers. By proceeding similarly, we can  show that the influence function for the functional $\sigma^2$ is  bounded. We refer to \cite{ghosh2013robust} for more details.

\section{Numerical Results} \label{sec:simulations}
We conducted an extensive simulation study to evaluate the performance of our proposed methodology. For comparison, we include the estimators from the following methods: ordinary least squares (OLS), LASSO, LAD LASSO \cite{wang2007robust}, GAM \cite{mgcvbook}, GAMSEL \cite{chouldechova2015generalized}, and RPLAM \cite{boente2023robust}. RPLAM did not converge in several setups, so we have excluded those cases. Except for OLS and GAM, the remaining approaches produce sparse coefficients. The programming language \texttt{R} \cite{Rlanguage} is used to conduct the simulations. 

We generate data from the additive model \eqref{eqn:model-def}. We consider the sample size, $n=250$, and the number of predictors, $p$, as 15 and 20. First, the covariates $X_j$'s are generated from a $p$-variate normal distribution $N_p(0, \Sigma)$, where the $(i, j)$-th element of $\Sigma$ is computed as $\sigma_{i, j} = (0.5)^{|i-j|}$. Each covariate is now transformed back to the range $(0,1)$ using the cumulative distribution function of the standard normal distribution. The covariates still exhibit correlation but with different magnitudes. 

The true model contains eight nonzero additive functions, as given below:
\begin{align*}
    g_{j_1} (x) &= a_{j_1} \sin(2\pi x), \ \ g_{j_2} (x) = a_{j_2} \sin(2\pi x),\\
    g_{j_3} (x) &= a_{j_3} \cos(2\pi x), \ \ g_{j_4} (x) = a_{j_4} \cos(2\pi x),\\
    g_{j_5} (x) &= a_{j_5} \exp( x), \ \ g_{j_6} (x) = a_{j_6} \exp( x),\\
    g_{j_7} (x) &= a_{j_7}  x, \ \ g_{j_8} (x) = a_{j_8}  x.
\end{align*}
The remaining $(p-8)$ additive components are set to zero. Here $j_1, j_2, \cdots, j_8$ is a random permutation of $(1, 2, \ldots, p)$, which means that the nonzero functions are placed randomly over the components $X_j$. The coefficients $a_{j_1}, a_{j_2}, \cdots, a_{j_8}$ are also random -- half of them are generated from $U(1, 2)$, and the remaining half are from $U(-2, -1)$. We deliberately avoided these coefficients from the neighborhood of zero to make sure the corresponding additive functions were far from a zero function. 

We consider five $(m_n=5)$ B-spline basis functions for each covariate $X_j$ using the \texttt{R} package \texttt{fda} \cite{fdapackage}. We use the group SCAD penalty, where the regularization parameter ($\lambda_n$) is selected using the AIC, BIC, EBIC, and Cp statistics. After selecting the regularization parameter based on different information criteria, we obtain the optimum DPD parameter $\nu$ that minimizes the H-score in Equation \eqref{h_score}. Thus, we have four estimators with the optimum $\lambda_n$ and $\nu$, which are denoted as AIC(DPD), BIC(DPD), EBIC(DPD), and Cp(DPD).

\begin{table*}[t]
\setlength{\tabcolsep}{1.5pt}
\caption{The average RPE, MSE of $\widehat{\sigma}$, sensitivity, and specificity of different estimators in pure data with $\epsilon \sim N(0,\sigma^2)$ and contaminated data with 5\% contamination from $\epsilon_c \sim N(8, \sigma^2)$, where $\sigma=1$ (SNR=$5.2763$). The second line gives the corresponding standard deviations (SD).  The number of additive components, $p$, is 15, with 7 set to zero. }
\label{tab:case1}
\centering
\scriptsize
\begin{tabularx}{\textwidth}{lcccccccc}
 \toprule
\multirow{2}{*}{\vspace{-6pt}\textbf{Estimators}} & \multicolumn{4}{c}{\textbf{Pure Data}}  & \multicolumn{4}{c}{\textbf{Contaminated Data}} \\ 
\cmidrule{2-9}
 & \textbf{RPE} & \textbf{MSE($\widehat{\sigma})$} & \textbf{Sensitivity} & \textbf{Specificity} & \textbf{RPE} & \textbf{MSE($\widehat{\sigma})$} & \textbf{Sensitivity} & \textbf{Specificity}\\
  \midrule
    OLS & 4.0458 & 0.7649 & -- & -- & 4.0791 & 2.2625 & -- & -- \\ 
  \bf{(SD)} & 0.0771 & 0.1003 & -- & -- & 0.2841 & 0.6234 & -- & -- \\ 
  LASSO & 4.0983 & 0.7669 & 0.8662 & 0.4686 & 3.9906 & 2.2609 & 0.8562 & 0.4386 \\ 
  \bf{(SD)} & 0.0854 & 0.1021 & 0.0857 & 0.2814 & 0.2560 & 0.6222 & 0.1026 & 0.2538 \\ 
  LAD LASSO & 6.3820 & 1.6690 & 0.3850 & 0.8529 & 5.3164 & 3.2225 & 0.3538 & \bf{0.9714} \\ 
  \bf{(SD)} & 0.2297 & 0.1768 & 0.0553 & 0.0554 & 0.2440 & 0.7458 & 0.0834 & 0.0673 \\ 
  GAM & 1.4983 & 0.0220 & -- & -- & 3.3377 & 0.5646 & -- & -- \\ 
  \bf{(SD)} & 0.1005 & 0.0134 & -- & -- & 0.7791 & 0.3016 & -- & -- \\ 
  GAMSEL & 1.2660 & \bf{0.0053} & \bf{1.0000} & 0.1257 & 2.0432 & 0.8613 & \bf{1.0000} & 0.1400 \\ 
  \bf{(SD)} & 0.0544 & 0.0059 & 0.0000 & 0.1305 & 0.3492 & 0.4236 & 0.0000 & 0.1559 \\ 
  AIC(DPD) & 1.2529 & 0.0064 & 1.0000 & 0.8118 & 1.5512 & 0.0127 & 0.9950 & 0.7962 \\ 
  \bf{(SD)} & 0.0729 & 0.0078 & 0.0000 & 0.1865 & 0.1919 & 0.0304 & 0.0321 & 0.2907 \\ 
  BIC(DPD) & 1.2854 & 0.0062 & 0.9802 & 0.9161 & \bf{1.5302} & 0.0126 & 0.9933 & 0.9200 \\ 
 \bf{(SD)}  & 0.1402 & 0.0108 & 0.0717 & 0.1046 & 0.1783 & 0.0306 & 0.0350 & 0.1508 \\ 
  EBIC(DPD) & 1.3871 & 0.0088 & \bf{0.9345} & 0.9705 & 1.5306 & 0.0133 & 0.9917 & 0.9448 \\ 
 \bf{(SD)} & 0.2287 & 0.0132 & 0.1050 & 0.0637 & 0.1820 & 0.0309 & 0.0375 & 0.0843 \\ 
  CP(DPD) & \bf{1.2529} & 0.0061 & 1.0000 & 0.8322 & 1.5368 & \bf{0.0116} & 0.9950 & 0.8152 \\ 
 \bf{(SD)} & 0.0725 & 0.0074 & 0.0000 & 0.1673 & 0.1826 & 0.0280 & 0.0321 & 0.2613 \\ 
   \bottomrule
\end{tabularx}

\end{table*}

\begin{table*}[]
\setlength{\tabcolsep}{1.5pt}
\caption{The average RPE, MSE of $\widehat{\sigma}$, sensitivity, and specificity of different estimators in pure data with $\epsilon \sim N(0,\sigma^2)$ and data with 5\% contamination with outliers drawn from $\epsilon_c \sim N(8, \sigma^2)$, where $\sigma=1$ (SNR=$7.8082$). The second line gives the corresponding standard deviations (SD). The number of additive components, $p$, is 20, with 12 set to zero.}
\label{tab:case2}
\centering
\scriptsize
\begin{tabularx}{\textwidth}{lcccccccc}
 \toprule
\multirow{2}{*}{\vspace{-6pt}\textbf{Estimators}} & \multicolumn{4}{c}{\textbf{Pure Data}}  & \multicolumn{4}{c}{\textbf{Contaminated Data}} \\ 
\cmidrule{2-9}
 & \textbf{RPE} & \textbf{MSE($\widehat{\sigma})$} & \textbf{Sensitivity} & \textbf{Specificity} & \textbf{RPE} & \textbf{MSE($\widehat{\sigma})$} & \textbf{Sensitivity} & \textbf{Specificity}\\
  \midrule
   OLS & 5.2655 & 1.6655 & -- & -- & 5.5573 & 3.5274 & -- & -- \\ 
  \bf{(SD)} & 0.1081 & 0.1506 & -- & -- & 0.3323 & 0.6980 & -- & -- \\ 
  LASSO & 4.9658 & 1.6753 & 0.8975 & 0.3500 & 5.1435 & 3.5499 & 0.8662 & 0.4908 \\ 
  \bf{(SD)} & 0.1121 & 0.1515 & 0.0544 & 0.1800 & 0.2728 & 0.7128 & 0.0993 & 0.2336 \\ 
  LAD LASSO & 5.1596 & 1.8810 & 0.7425 & 0.9842 & 5.1214 & 3.8940 & 0.7000 & \bf{0.9867} \\ 
  \bf{(SD)} & 0.1920 & 0.1835 & 0.0867 & 0.0349 & 0.2318 & 0.7654 & 0.0852 & 0.0307 \\ 
  GAM & 1.6822 & 0.0412 & -- & -- & 3.7573 & 0.4145 & -- & -- \\ 
  \bf{(SD)} & 0.1539 & 0.0164 & -- & -- & 0.7688 & 0.2378 & -- & -- \\ 
  GAMSEL & \bf{1.2117} & 0.0090 & \bf{1.0000} & 0.1275 & 1.6969 & 0.7997 & \bf{0.9888} & 0.2117 \\ 
  \bf{(SD)} & 0.0610 & 0.0080 & 0.0000 & 0.1089 & 0.2516 & 0.3603 & 0.0360 & 0.1337 \\ 
  AIC(DPD) & 1.2319 & 0.0081 & 0.9984 & 0.8323 & 1.5444 & 0.0444 & 0.9404 & 0.8643 \\ 
  \bf{(SD)} & 0.1205 & 0.0086 & 0.0140 & 0.2240 & 0.2002 & 0.0793 & 0.1063 & 0.1918 \\ 
  BIC(DPD) & 1.2836 & 0.0067 & 0.9359 & 0.9719 & 1.5695 & 0.0497 & 0.8808 & 0.9622 \\ 
  \bf{(SD)} & 0.1558 & 0.0077 & 0.1125 & 0.0439 & 0.2169 & 0.0775 & 0.1299 & 0.0928 \\ 
  EBIC(DPD) & 1.3911 & 0.0108 & 0.8469 & 0.9896 & 1.5933 & 0.0567 & 0.8503 & 0.9797 \\ 
  \bf{(SD)} & 0.1932 & 0.0130 & 0.1289 & 0.0277 & 0.2073 & 0.0779 & 0.1284 & 0.0423 \\ 
  CP(DPD) & 1.2201 & \bf{0.0058} & 0.9984 & 0.8927 & \bf{1.5400} & \bf{0.0431} & 0.9273 & 0.8953 \\ 
  \bf{(SD)} & 0.0875 & 0.0064 & 0.0140 & 0.1300 & 0.1994 & 0.0786 & 0.1109 & 0.1638 \\ 
  \bottomrule
\end{tabularx}

\end{table*}

The performance of the estimators is compared using four measures: the relative prediction error (RPE), the mean squared error (MSE) of $\widehat{\sigma}$, sensitivity, and specificity. To calculate the RPE, we have generated test data of size $R=1000$. Then, the RPE is defined as
\begin{equation}
RPE = \frac{1}{R  \sigma^2} \sum_{r=1}^R \left(\widehat{y}_r - y_r\right)^2, 
\label{rpe}
\end{equation}
where $y_r$ is the $r$-th observation in the test data, and $\widehat{y}_r$ is its predicted value calculated from the training data. Here, $\sigma^2$ is the true value of the model variance. We replicated each simulation using $S=100$ training datasets, and their average performance is reported. We compared the performance of the estimators both for the pure and the contaminated data. However, as the outliers are unpredictable, the test data is always generated from the pure model without contamination to avoid prediction being affected by the outliers.  Otherwise, if the test data were generated from the contaminated distribution $m$, defined in Section \ref{sec:dpd}, the RPE serves as a performance metric to assess how closely the estimated distribution \( f_{\widehat{\theta}} \) aligns with \( m \). Since our aim is not to model the outliers, we seek a measure that reflects the deviance between \( f_{\theta_0} \) and \( f_{\widehat{\theta}} \). This can be derived from the RPE, but it necessitates that the test data be generated directly from \( f_{\theta_0} \).

The MSE of $\widehat{\sigma}$ is calculated as
\begin{equation}
MSE = \frac{1}{S} \sum_{s=1}^S \left( \widehat{\sigma}^s-\sigma \right)^2,
\end{equation}
where $\widehat{\sigma}^s$ is the estimate of $\sigma$ obtained from the $s$-th replication of the training data. The other two measures, sensitivity and specificity, are also calculated from the training data, and the average values are reported in the following tables. While sensitivity is the proportion of correctly estimated nonzero additive components, specificity is the proportion of correctly estimated zero additive components in model \eqref{eqn:model-def}.

First, we generated pure data, where the error term in model \eqref{eqn:model-def} comes from $\epsilon \sim N(0, \sigma^2)$. We have taken $\sigma=1$, and the number of additive components, $p$, is 15, in which 7 additive functions ($g_i$) are set to zero. The simulated signal-to-ratio (SNR) is approximately $6.9645$. The results are reported in Table \ref{tab:case1}, where the best results are displayed in bold font. As the true model is non-linear and sparse, the three estimators, OLS, LASSO, and LAD LASSO, which assume a linear relationship, give very high RPE values. On the other hand, the remaining estimators, produce low RPE values. The MSE of $\widehat{\sigma}$ is also very low. All our proposed estimators and GAMSEL give very high sensitivity values. Furthermore, our estimators produce high specificity values, whereas the value of specificity from the GAMSEL is only $0.15$. Thus, GAMSEL does not reduce model dimension enough in this setup.  

Using the same setup, we compared the estimators in contaminated data, which includes 5\% outliers drawn from $N(8, \sigma^2)$ in the error term. As it violates the assumption that $\mathbb{E}(\epsilon)=0$, the intercept parameter of the model is unidentifiable. However, it does not impact the other parameters or the predicted values. We notice that the performance of GAM and GAMSEL estimators is deteriorated, as they give inflated values of the RPE and the MSE of $\widehat{\sigma}$. The values for the proposed DPD estimators are almost the same as those obtained in the pure data. It shows that the proposed method successfully eliminated the effect of outliers. Although the RPLAM approach does not provide sparse estimators, it produces robust estimators. However, it comes with a high price by losing efficiency in pure data as RPE and MSE values were large. It is a common feature with many robust estimators -- they lose significant efficiency to achieve robustness. On the other hand, the tuning parameter $\nu$ in our DPD measure controls the efficiency and robustness. Moreover, the adaptive choice of $\nu$ successfully produces high efficiency in pure data and gives proper robustness property when there are outliers. 


Now, we consider a relatively high dimension by taking $p=20$ covariates. It gives a simulated SNR of 7.8082. The results are presented in Table \ref{tab:case2}. For the contaminated data, we consider contamination of 5\% with outliers drawn from $N(8, \sigma^2)$, where $\sigma = 1$. Table \ref{tab:case1} shows that the results are similar to the previous simulation.

Finally, we consider two heavy-tailed error distributions -- the Cauchy and the $\chi^2_1$, for the contaminated data. Theoretically, $\mathbb{E}(\epsilon)$ must be bounded to allow for the adjustment of the model's intercept term, centering $\mathbb{E}(\epsilon)$ at zero. While the Cauchy distribution has an unbounded mean, the simulated data remains finite. Therefore, we selected this example to illustrate and compare the estimators in the presence of extreme outliers. Since we do not have any target $\sigma$ value for the error term, the MSE of $\widehat{\sigma}$ is not computed. Moreover, the RPE measure is replaced by the mean prediction error (MPE) that does not have the $\sigma^2$ term in the denominator of the RPE in \eqref{rpe}. We generated $R=1000$ test data from the same distribution (Cauchy or $\chi^2_1$) to calculate the MPE. It should be noted that the test data also contain outliers with respect to the normal model, so the MPE is not a good measure to compare estimators that give equal weight (with the square loss) to all observations. One alternative is to delete a portion (say, $\omega$\%) of observations that yield large residuals (on either side), then calculate the MPE based on the remaining $(100-\omega)$\% test data. We can take a reasonable value of $\omega$ (say, $\omega=5$), provided it is not too small or too large. Table \ref{tab:case3} displays the results, where we used the full MPE and the MPE with 5\% trimmed test data. As expected, the full MPE is very high for all estimators, as they also contain predictions for the outlying observations. On the other hand, the 5\% trimmed MPE of the RPLAM and our methods are very small.  The other estimators completely break down in the case of the Cauchy distribution, as both the mean and SD of the MPEs are exceptionally large. Notably, the SDs of MPEs are approximately ten times larger than the means, highlighting the sensitivity of the SD to outliers due to its dependence on the square of the differences from the mean. In contrast, the $\chi^2_1$ distribution has a mild influence on non-robust estimators, as its tail is relatively lighter.

Theoretically, the Cauchy error violates the model assumption in \eqref{eqn:model-def} as its variance is unbounded. However, in the simulation, we can make it bounded by setting a big bound on the observations. 
Overall, the simulation results illustrate that the proposed estimators yield superior performance in the presence of contaminated data while achieving model selection. At the same time, unlike other robust estimators, they do not lose much efficiency in the pure data.  
\begin{table*}[t]
\setlength{\tabcolsep}{1pt}
\caption{The average MPE in 5\% trimmed and full data, sensitivity, and specificity of different estimators for the Cauchy and $\chi^2_1$ distributions. The second line gives the corresponding standard deviations (SD). While the number of additive components, $p$, is 20, with 12 set to zero in the first case,  the number of additive components, $p$, is 15, with 7 set to zero in the second case.}
\label{tab:case3}
\centering
\scriptsize
\begin{tabularx}{\textwidth}{lcccccccc}
 \toprule
\multirow{2}{*}{\vspace{-6pt}\textbf{Estimators}} & \multicolumn{4}{c}{\textbf{Cauchy Distribution}}  & \multicolumn{4}{c}{\textbf{$\chi^2_1$ Distribution}} \\ 
\cmidrule{2-9}
 & \textbf{MPE} & \textbf{MPE} & \textbf{Sensitivity} & \textbf{Specificity} & \textbf{MPE} & \textbf{MPE} & \textbf{Sensitivity} & \textbf{Specificity}\\
 & \textbf{(5\%)} & \textbf{(full)} & \textbf{ } & \textbf{ } & \textbf{(5\%)} & \textbf{(full)} & \textbf{ } & \textbf{ }\\
  \midrule
   OLS & 40921.21 & 82108.87 & -- & -- & 3.09 & 4.37 & -- & -- \\ 
  \bf{(SD)} & 393346.61 & 706828.95 & -- & -- & 0.10 & 0.10 & -- & -- \\ 
  LASSO & 1878.79 & 4076.14 & 0.17 & 0.91 & 3.04 & 4.32 & 0.91 & 0.42 \\ 
  \bf{(SD)} & 17721.01 & 29432.14 & 0.23 & 0.16 & 0.09 & 0.09 & 0.09 & 0.22 \\ 
  LADLASSO & 89.73 & 886.83 & 0.57 & 0.94 & 4.89 & 6.72 & 0.32 & 0.94 \\ 
  \bf{(SD)} & 721.79 & 1486.03 & 0.11 & 0.12 & 0.21 & 0.29 & 0.16 & 0.08 \\ 
  GAM & 219624.01 & 850254.01 & -- & -- & 1.62 & 2.62 & -- & -- \\ 
  \bf{(SD)} & 2086581.73 & 7485942.04 & -- & -- & 0.27 & 0.30 & -- & -- \\ 
  GAMSEL & 1877.82 & 4075.41 & 0.15 & 0.94 & 1.16 & 2.09 & \bf{1.00} & 0.14 \\ 
  \bf{(SD)} & 17721.11 & 29432.23 & 0.26 & 0.13 & 0.13 & 0.13 & 0.00 & 0.14 \\ 
  RPLAM & 9.48 & 864.78 & -- & -- & \bf{0.99} & \bf{2.03} & -- & -- \\ 
  \bf{(SD)} & 0.62 & 940.30 & -- & -- & 0.51 & 0.15 & -- & -- \\ 
  AIC(DPD) & 9.43 & 716.43 & \bf{0.86} & 0.90 & 1.07 & 2.09 & 0.92 & 0.83 \\ 
  \bf{(SD)} & 0.61 & 3.04 & 0.10 & 0.17 & 0.16 & 0.14 & 0.11 & 0.22 \\ 
  BIC(DPD) & \bf{9.40} & 716.32 & 0.81 & 0.99 & 1.12 & 2.17 & 0.83 & 0.99 \\ 
  \bf{(SD)} & 0.52 & 3.04 & 0.09 & 0.04 & 0.16 & 0.14 & 0.10 & 0.04 \\ 
  EBIC(DPD) & 10.49 & \bf{715.63} & 0.62 & \bf{1.00} & 1.14 & 2.19 & 0.81 & \bf{1.00} \\ 
  \bf{(SD)} & 2.13 & 1.78 & 0.31 & 0.03 & 0.15 & 0.14 & 0.09 & 0.02 \\ 
  CP(DPD) & 9.43 & 716.40 & 0.85 & 0.91 & 1.07 & 2.10 & 0.91 & 0.84 \\ 
  \bf{(SD)} & 0.63 & 3.02 & 0.10 & 0.17 & 0.16 & 0.15 & 0.11 & 0.21 \\ 
  \bottomrule
\end{tabularx}

\end{table*}

\section{Real Data Analysis} \label{sec:data-analysis} 

This section presents two real data examples: NCI 60 cell line data \cite{shankavaram2007transcript} and forest fire data \cite{cortez2007data}. As discussed in Section \ref{sec:simulations}, the proposed estimators are compared with different methods. Since the RPLAM approach exhibited convergence issues in these datasets, it is not included in the further comparison. For comparison, we consider the root mean prediction error (RMPE) as given below
\begin{equation}
    RMPE = \sqrt{\frac{1}{n} \sum_{i=1}^n (y_i - \widehat{y}_i)^2},
\end{equation}
where $\widehat{y}_i$ is the estimate of the response variable, and $y_i$ is the original observation. 
 We have also used the mean absolute deviation (MAD) $\frac{1}{n} \sum_{i=1}^n |y_i - \widehat{y}_i|$ as a robust measure of prediction accuracy. We randomly partitioned the data into training and test sets, with 75\% observations belong to the training set. The process is repeated a few times and the average RMPE and MAD are reported. 
Finally, We computed the percentage of dimension reduction. Suppose that the final model selects $q$ regressor variables, where the number of variables in the full model is $p$. Then, we define the percentage of dimension reduction as 
\begin{equation}
    \text{Dim. Reduction} = \frac{100(p-q)}{p}. \label{dim_red}
\end{equation}

It is essential to recognize that RMPE is sensitive to outliers. Consequently, even with a robust model fit, a low RMPE may not be achieved if the test data contain outliers. While the GAM aims to minimize the residual sum of squares (RSS) within the training data, this can lead to a reduced test RMPE. However, in the process of minimizing RSS or RMPE, the GAM fit may gravitate toward the outliers. To obtain a more accurate representation from the RMPE measure, it is necessary to remove outliers from the test data. Nonetheless, detecting outliers can be a significant challenge and is often subjective, especially in high-dimensional spaces. Moreover, employing robust fitting techniques is crucial to mitigate the masking effect of these outliers.

To demonstrate that the robust estimator provides a better fit for the “good” observations, we first trimmed $\omega\%$ of the test observations with the highest absolute errors, $|y_i - \hat{y}_i|$, and subsequently calculated the RMPE from the remaining observations. Our objective is to illustrate that when all outliers are effectively removed from the test data, robust methods will yield a lower RMPE, even if their RMPE on the full test dataset may exceed that of OLS or GAM methods. Since the precise proportion of outliers is unknown, we explored several values of $\omega$, such as 5\% and 10\%. On the other hand, the MAD provides an overall performance metric that is less influenced by outliers. The LAD LASSO minimizes the MAD in the training data, which should lead to a smaller MAD in the test data under the assumption of a linear regression function. However, when the true regression function is non-linear, our proposed method is preferable.


\begin{table*}[]
\caption{The RMPE  and MAD for the NCI 60 Cell Line Data. The last row is the percentage of dimension reduction as given in Equation \eqref{dim_red} with $p=10$.\\}
\setlength{\tabcolsep}{1pt}
\scriptsize
\centering
\begin{tabular}{lrrrrrrrrr}
  \hline
   & & & & & & \multicolumn{4}{c}{DPD} \\
  \cline{7-10}
 & OLS & LASSO & LAD LASSO & GAM & GAMSEL & AIC & BIC & EBIC & Cp \\ 
  \hline
10\% Trim RMPE & 1.25 & 1.24 & 1.14 & 0.96 & 1.21 & \textbf{0.48} & 0.50 & 0.50 & 0.50 \\ 
  5\% Trim RMPE & 1.48 & 1.50 & 1.50 & 1.14 & 1.44 & \textbf{0.54} & 0.59 & 0.59 & 0.59 \\ 
  RMPE & 1.83 & 1.93 & 2.22 & \textbf{1.43} & 1.79 & 1.91 & 1.94 & 1.94 & 1.94 \\
  MAD & 1.35 & 1.40 & 1.42 & 1.04 & 1.35 & \textbf{0.70} &  0.74 & 0.74 & 0.74 \\
  Dim. Reduction & 0.00 & 40.00 & \textbf{70.00} & 0.00 & 10.00 & 0.00 & 10.00 & 10.00 & 10.00 \\ 
   \hline
\end{tabular} \label{data:NCI60}
\end{table*}

\bigskip
\textbf{NCI 60 Cell Line Data:}
The cancer cell panel dataset from the National Cancer Institute contains information on 60 human cancer cell lines. The dataset can be accessed through CellMiner, a web application available at http://discover.nci.nih.gov/cellminer/. Our analysis focuses on the relationship between protein expression and gene expression, where the latter is measured by an Affymetrix HG-U133A chip that provides data on 22,283 gene predictors, normalized using the GCRMA method. Based on 162 antibodies, protein expressions were obtained via reverse-phase protein lysate arrays and $\log_2$ transformed. We removed one observation due to missing values in the gene expression data, leading to a final sample size of $n = 59$. Our investigation builds upon previous research by \cite{shankavaram2007transcript,lee2011sparse}, who also used this dataset to study gene expression and protein levels but with non-robust methods that resulted in complex models with hundreds to thousands of predictors that are difficult to interpret. \cite{alfons2013sparse} analyzed this dataset using a robust sparse least trimmed square regression estimator. We aim to extend their analysis by comparing the performance of the DPD estimators to other competing methods.

The response variable is the protein expression based on the KRT18 antibody, constituting the variable with the largest mean absolute deviation (MAD). It measures the expression level of the protein keratin 18, which is known to be persistently expressed in carcinomas \cite{oshima1996oncogenic}. The sample size is very small, so we considered ten genes with the highest correlation with the response variable. We used a robust correlation estimate based on winsorization using the \emph{corHuber} function of the \texttt{robustHD} \cite{robusthd} package in \texttt{R}. 

Table \ref{data:NCI60} presents the RMPEs of different estimators for this dataset. As expected, while the GAM estimator shows the smallest RMPE in the full test data, the DPD estimator produces the smallest RMPE based on 5\% or 10\% truncated test data. For the other methods, however, there is only a marginal reduction of RMPE from the full to truncated test data, which indicates that the outliers significantly influence them. The MAD is also significantly low for our robust methods. In addition, we notice a low percentage of dimension reduction due to the high correlations of the ten genes with the protein, which range from 0.7760 to 0.6169. The AIC-based variable selection produces the full model, including all genes, and performs marginally better than the other DPD estimators.

\begin{figure}[t]
		\centering%
		\begin{tabular}{cccc} \hspace{-1cm}
			\includegraphics[width=3.5cm]{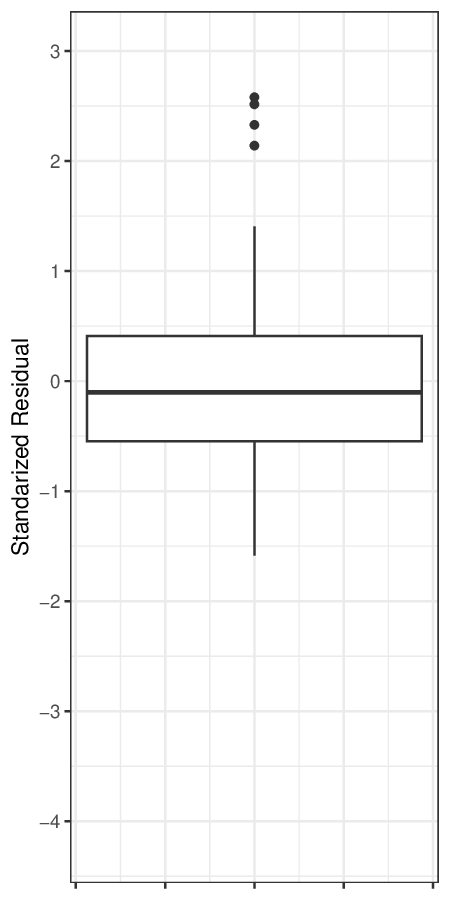} &
			\includegraphics[width=3.5cm]{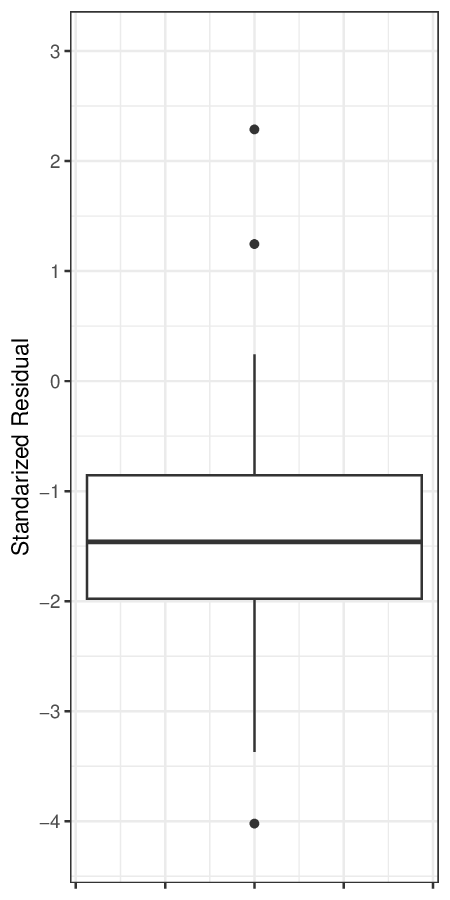} & \includegraphics[width=3.5cm]{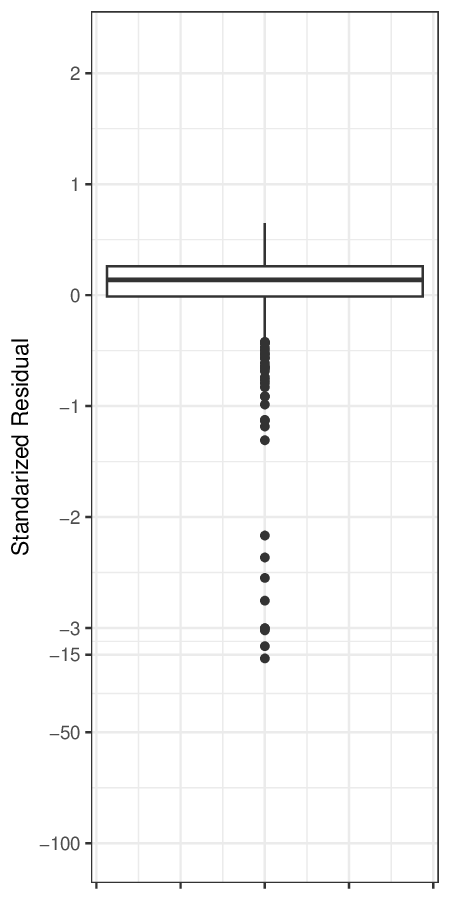} &
			\includegraphics[width=3.5cm]{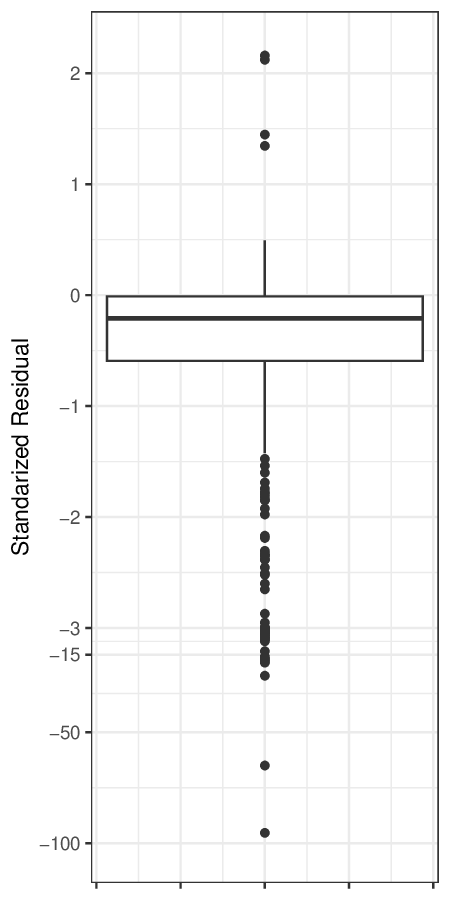}             \\
            OLS (NCI Data) & DPD (NCI Data) & OLS (Forest Fire  Data) & DPD (Forest Fire  Data)
			 \end{tabular}%
		\caption{The residual box plots in the complete NCI and forest fire data using the OLS and DPD with $\nu = 0.05$.}
		\label{fig:boxplot}
	\end{figure}

The first two box plots in Figure \ref{fig:boxplot} illustrate the distributions of standardized residuals obtained from fitting the complete dataset using OLS and the DPD estimator with $\nu = 0.05$. While the box plots indicate a few mild outliers, all standardized residuals (with the exception of a few from the DPD estimator) fall within the range of $-3$ to 3. Therefore, there is generally no strong evidence to warrant the removal of any observations. However, if we were to exclude the four observations with OLS standardized residuals exceeding 2, the RMPE and MAD for OLS would significantly decrease from $1.83$ and $1.35$ to $0.92$ and $0.71$, respectively. Notably, the DPD method demonstrates even better performance without requiring the deletion of any observations.

\bigskip
\noindent
\textbf{Forest Fire Data:} We investigate the forest fire data collected from the Montesinho natural park in the Trás-os-Montes northeast region of Portugal. The park is known for its diverse flora and fauna and is situated in a supraediterranean climate with an average annual temperature ranging from 8 to 12\textdegree C. Our analysis is based on data collected from two main sources between January 2000 and December 2003. The first source is a database compiled by the inspector monitoring fire occurrences in Montesinho. This database includes information on several features related to each forest fire event, such as the date, time, and spatial location within a $9\times 9$
grid ($x$ and $y$ axis), the type of vegetation involved, the four components of the Fire Weather Index (FWI) system, and the total burned area. The second source of data is collected by a meteorological station located in the center of Montesinho Park and managed by the Bragança Polytechnic Institute. This database contains weather conditions, including temperature, relative humidity, wind speed, and rainfall. The combined dataset contains 517 entries with 13 attributes. We refer to \cite{cortez2007data} for more information about this dataset.

\begin{table*}[]
\caption{The RMPE and MAD for the forest fire data. The last row is the percentage of dimension reduction as given in Equation \eqref{dim_red} with $p=11$.\\}
\setlength{\tabcolsep}{1pt}
\scriptsize
\centering
\begin{tabular}{lrrrrrrrrr}
    \hline
   & & & & & & \multicolumn{4}{c}{DPD} \\
  \cline{7-10}
 & OLS & LASSO & LAD LASSO & GAM & GAMSEL & AIC & BIC & EBIC & Cp \\ 
  \hline
10\% Trim RMPE & 13.42 & 13.21 & 5.21 & 14.52 & 11.13 & \textbf{4.42} & 4.50 & 4.50 & 4.50 \\ 
  5\% Trim RMPE & 14.93 & 12.01 & 9.36 & 16.55 & 11.98 & \textbf{8.48} & 8.54 & 8.54 & 8.54 \\ 
  RMPE & 62.81 & 62.82 & 64.77 & \textbf{61.67} & 63.59 & 64.36 & 64.40 & 64.40 & 64.40 \\ 
  MAD & 19.54 & 18.93 & \textbf{12.89} & 19.89 & 18.57 & 12.93 &  13.00 & 13.00 & 13.00 \\
  Dim. Reduction & 0.00 & 0.00 & 90.91 & 0.00 & \textbf{100.00} & 18.18 & 54.55 & 54.55 & 54.55 \\ 
   \hline
\end{tabular} \label{data:fire}
\end{table*}

The burned area is the response variable for this dataset. However, it includes 247 samples with a zero value, indicating a burn area lower than 100 square meters. The RMPE values of different estimators are presented in Table \ref{data:fire}. In the full test data, GAM gives the smallest RMPE. In fact, all methods give very high RMPE in the full test data. However, we noticed a substantial reduction of RMPE in 5\% and 10\% trimmed test data. The LAD LASSO produces the smallest MAD value, indicating that the true regression model is almost linear. The LAD LASSO and DPD-based estimators give the best results for this dataset. It indicates that they provide an excellent fit for ``good'' observations and down-weights outliers. On the other hand, other methods lose significant efficiency due to outlying observations. 

For the DPD-based estimators, all variable selection methods, except the AIC, provide the same outputs. The AIC is liberal in reducing the dimension of the covariates but produces the lowest RMPE in the trimmed test data. It is interesting to see that the GAMSEL approach eliminates all covariates and yields the intercept-only model. Nevertheless, it gives better prediction error than the OLS, LASSO, and GAM estimators in the trimmed test data. Furthermore, DPD-based estimators show that further reduction of RMPE values can be achieved with the help of properly selected covariates.

The last two box plots in Figure \ref{fig:boxplot} display the standardized residuals for the complete dataset, calculated using both the OLS and DPD estimator with $\nu = 0.05$. Among the OLS-standardized residuals, only five values fall below $-3$. In contrast, the DPD plot reveals that 36 observations, constituting 6.96\% of the data, occupy this lower region. Should we eliminate these 36 observations, the RMPE and MAD for OLS would decline dramatically from 62.81 and 19.54 to 6.71 and 4.52, respectively. The presence of outliers results in an overestimation of $\sigma$ in the OLS method (calculated as 63.55), effectively masking many outliers in the residual plot. In contrast, the DPD method inherently downweights outliers, allowing for improved performance without explicit detection.

\section{Summary and Conclusions}
This study provides a new robust variable selection method in the high dimensional additive models. The proposed approach estimates the model using B-splines with density power divergence loss function and employs a nonconcave penalty for the regularization. Our theoretical results are based on the assumption that the error is sub-Weibull, which includes both sub-Exponential and sub-Gaussian as special cases. The results from the numerical study illustrate that the proposed method works well with finite samples. While we consider multiple information criteria for model selection, we observe that they all produce similar results, and none stand out. We noticed that the H-score function is unstable when the total number of basis functions is large compared to the sample size. It is due to the small value of $\widehat{\sigma}$ in the over-parameterized case, and the H-score function fails to find the optimum value of the DPD parameter $\nu$. Alternatively, we can estimate $\sigma$ following \cite{chen2018error} and update the H-score function. However, this procedure is computationally intensive as it uses a model selection and cross-validation technique iteratively. Therefore, finding an efficient algorithm in this situation is a worthwhile topic for future research.  Another interesting direction for future research is to consider the robust variable selection in high dimensional models using other divergence measures, e.g., \cite{jones2001comparison} and \cite{maji2019robust}.

\section*{Supplementary Materials}

Supplementary material includes the necessary lemma and proofs of the asymptotic results in Section \ref{sec:asymp}.

\section*{Acknowledgements}
The authors thank the associate editor and the two reviewers for their valuable suggestions which led to significant improvements in the paper.

\bibliographystyle{IEEEtran}      
\bibliography{highdpd}   

\begin{thebibliography}{10}
\providecommand{\url}[1]{#1}
\csname url@samestyle\endcsname
\providecommand{\newblock}{\relax}
\providecommand{\bibinfo}[2]{#2}
\providecommand{\BIBentrySTDinterwordspacing}{\spaceskip=0pt\relax}
\providecommand{\BIBentryALTinterwordstretchfactor}{4}
\providecommand{\BIBentryALTinterwordspacing}{\spaceskip=\fontdimen2\font plus
\BIBentryALTinterwordstretchfactor\fontdimen3\font minus
  \fontdimen4\font\relax}
\providecommand{\BIBforeignlanguage}[2]{{%
\expandafter\ifx\csname l@#1\endcsname\relax
\typeout{** WARNING: IEEEtran.bst: No hyphenation pattern has been}%
\typeout{** loaded for the language `#1'. Using the pattern for}%
\typeout{** the default language instead.}%
\else
\language=\csname l@#1\endcsname
\fi
#2}}
\providecommand{\BIBdecl}{\relax}
\BIBdecl

\bibitem{friedman1981projection}
J.~H. Friedman and W.~Stuetzle, ``Projection pursuit regression,''
  \emph{Journal of the American Statistical Association}, vol.~76, no. 376, pp.
  817--823, 1981.

\bibitem{stone1985additive}
C.~J. Stone, ``Additive regression and other nonparametric models,''
  \emph{Annals of Statistics}, vol.~13, no.~2, pp. 689--705, 1985.

\bibitem{hastie1986gam}
\BIBentryALTinterwordspacing
T.~Hastie and R.~Tibshirani, ``{Generalized Additive Models},''
  \emph{Statistical Science}, vol.~1, no.~3, pp. 297 -- 310, 1986. [Online].
  Available: \url{https://doi.org/10.1214/ss/1177013604}
\BIBentrySTDinterwordspacing

\bibitem{buja1989linear}
A.~Buja, T.~Hastie, and R.~Tibshirani, ``Linear smoothers and additive
  models,'' \emph{Annals of Statistics}, vol.~17, no.~2, pp. 453--510, 1989.

\bibitem{wood2006generalized}
S.~Wood, \emph{Generalized Additive Models: An Introduction with R},
  2nd~ed.\hskip 1em plus 0.5em minus 0.4em\relax Chapman and Hall/CRC, 2017.

\bibitem{mammen1999existence}
E.~Mammen, O.~Linton, and J.~Nielsen, ``The existence and asymptotic properties
  of a backfitting projection algorithm under weak conditions,'' \emph{Annals
  of Statistics}, vol.~27, no.~5, pp. 1443--1490, 1999.

\bibitem{linton1995kernel}
O.~Linton and J.~P. Nielsen, ``A kernel method of estimating structured
  nonparametric regression based on marginal integration,'' \emph{Biometrika},
  vol.~82, no.~1, pp. 93--100, 1995.

\bibitem{lin2006component}
Y.~Lin and H.~H. Zhang, ``Component selection and smoothing in multivariate
  nonparametric regression,'' \emph{Annals of Statistics}, vol.~34, no.~5, pp.
  2272--2297, 2006.

\bibitem{meier2009high}
L.~Meier, S.~Van~de Geer, and P.~B{\"u}hlmann, ``High-dimensional additive
  modeling,'' \emph{Annals of Statistics}, vol.~37, no.~6B, pp. 3779--3821,
  2009.

\bibitem{ravikumar2009sparse}
P.~Ravikumar, J.~Lafferty, H.~Liu, and L.~Wasserman, ``Sparse additive
  models,'' \emph{Journal of the Royal Statistical Society: Series B
  (Statistical Methodology)}, vol.~71, no.~5, pp. 1009--1030, 2009.

\bibitem{huang2010variable}
J.~Huang, J.~L. Horowitz, and F.~Wei, ``Variable selection in nonparametric
  additive models,'' \emph{Annals of Statistics}, vol.~38, no.~4, pp.
  2282--2313, 2010.

\bibitem{xue2006additive}
L.~Xue and L.~Yang, ``Additive coefficient modeling via polynomial spline,''
  \emph{Statistica Sinica}, vol.~16, no.~4, pp. 1423--1446, 2006.

\bibitem{xue2009consistent}
L.~Xue, ``Consistent variable selection in additive models,'' \emph{Statistica
  Sinica}, vol.~19, no.~3, pp. 1281--1296, 2009.

\bibitem{wang2007group}
L.~Wang, G.~Chen, and H.~Li, ``Group scad regression analysis for microarray
  time course gene expression data,'' \emph{Bioinformatics}, vol.~23, no.~12,
  pp. 1486--1494, 2007.

\bibitem{wang2008variable}
L.~Wang, H.~Li, and J.~Z. Huang, ``Variable selection in nonparametric
  varying-coefficient models for analysis of repeated measurements,''
  \emph{Journal of the American Statistical Association}, vol. 103, no. 484,
  pp. 1556--1569, 2008.

\bibitem{huang2012selective}
J.~Huang, P.~Breheny, and S.~Ma, ``A selective review of group selection in
  high-dimensional models,'' \emph{Statistical Science}, vol.~27, no.~4, pp.
  481 -- 499, 2012.

\bibitem{selesnick2017sparse}
I.~Selesnick, ``Sparse regularization via convex analysis,'' \emph{IEEE
  Transactions on Signal Processing}, vol.~65, no.~17, pp. 4481--4494, 2017.

\bibitem{liu2021convex}
X.~Liu, A.~J. Molstad, and E.~C. Chi, ``A convex-nonconvex strategy for grouped
  variable selection,'' 2021, arXiv preprint arXiv:2111.15075.

\bibitem{chouldechova2015generalized}
A.~Chouldechova and T.~Hastie, ``Generalized additive model selection,'' 2015,
  arXiv preprint arXiv:1506.03850.

\bibitem{amato2016additive}
U.~Amato, A.~Antoniadis, and I.~De~Feis, ``Additive model selection,''
  \emph{Statistical Methods \& Applications}, vol.~25, no.~4, pp. 519--564,
  2016.

\bibitem{sadhanala2019additive}
V.~Sadhanala and R.~J. Tibshirani, ``Additive models with trend filtering,''
  \emph{Annals of Statistics}, vol.~47, no.~6, pp. 3032--3068, 2019.

\bibitem{bianco1998robust}
A.~Bianco and G.~Boente, ``Robust kernel estimators for additive models with
  dependent observations,'' \emph{The Canadian Journal of Statistics/La Revue
  Canadienne de Statistique}, vol.~26, no.~2, pp. 239--255, 1998.

\bibitem{alimadad2011outlier}
A.~Alimadad and M.~Salibian-Barrera, ``An outlier-robust fit for generalized
  additive models with applications to disease outbreak detection,''
  \emph{Journal of the American Statistical Association}, vol. 106, no. 494,
  pp. 719--731, 2011.

\bibitem{croux2012robust}
C.~Croux, I.~Gijbels, and I.~Prosdocimi, ``Robust estimation of mean and
  dispersion functions in extended generalized additive models,''
  \emph{Biometrics}, vol.~68, no.~1, pp. 31--44, 2012.

\bibitem{wong2014robust}
R.~K. Wong, F.~Yao, and T.~C. Lee, ``Robust estimation for generalized additive
  models,'' \emph{Journal of Computational and Graphical Statistics}, vol.~23,
  no.~1, pp. 270--289, 2014.

\bibitem{boente2017marginal}
G.~Boente and A.~Mart{\'\i}nez, ``Marginal integration m-estimators for
  additive models,'' \emph{Test}, vol.~26, no.~2, pp. 231--260, 2017.

\bibitem{boente2017robust}
G.~Boente, A.~Mart{\'\i}nez, and M.~Salibi{\'a}n-Barrera, ``Robust estimators
  for additive models using backfitting,'' \emph{Journal of Nonparametric
  Statistics}, vol.~29, no.~4, pp. 744--767, 2017.

\bibitem{boente2023robust}
G.~Boente and A.~M. Mart{\'\i}nez, ``A robust spline approach in partially
  linear additive models,'' \emph{Computational Statistics \& Data Analysis},
  vol. 178, p. 107611, 2023.

\bibitem{amato2022wavelet}
U.~Amato, A.~Antoniadis, I.~D. Feis, and I.~Gijbels, ``Wavelet-based robust
  estimation and variable selection in nonparametric additive models,''
  \emph{Statistics and Computing}, vol.~32, no.~1, pp. 1--19, 2022.

\bibitem{MR1665873}
A.~Basu, I.~R. Harris, N.~L. Hjort, and M.~C. Jones, ``Robust and efficient
  estimation by minimising a density power divergence,'' \emph{Biometrika},
  vol.~85, no.~3, pp. 549--559, 1998.

\bibitem{mandal2019robust}
A.~Mandal and S.~Ghosh, ``Robust variable selection criteria for the penalized
  regression,'' 2019, arXiv preprint arXiv:1912.12550.

\bibitem{mandal2022robust}
------, ``Robust lasso and its applications in healthcare data,'' \emph{Trends
  in Mathematical, Information and Data Sciences: A Tribute to Leandro Pardo},
  vol. 445, pp. 389--398, 2022.

\bibitem{lv2009unified}
J.~Lv and Y.~Fan, ``A unified approach to model selection and sparse recovery
  using regularized least squares,'' \emph{Annals of Statistics}, vol.~37,
  no.~6A, pp. 3498--3528, 2009.

\bibitem{fujisawa2008robust}
H.~Fujisawa and S.~Eguchi, ``Robust parameter estimation with a small bias
  against heavy contamination,'' \emph{Journal of Multivariate Analysis},
  vol.~99, no.~9, pp. 2053--2081, 2008.

\bibitem{das2022testing}
J.~Das, B.~H. Beyaztas, M.~K. Mac-Ocloo, A.~Majumdar, and A.~Mandal, ``Testing
  equality of multiple population means under contaminated normal model using
  the density power divergence,'' \emph{Entropy}, vol.~24, no.~9, p. 1189,
  2022.

\bibitem{frank1993statistical}
L.~E. Frank and J.~H. Friedman, ``A statistical view of some chemometrics
  regression tools,'' \emph{Technometrics}, vol.~35, no.~2, pp. 109--135, 1993.

\bibitem{tibshirani1996regression}
R.~Tibshirani, ``Regression shrinkage and selection via the lasso,''
  \emph{Journal of the Royal Statistical Society: Series B (Methodological)},
  vol.~58, no.~1, pp. 267--288, 1996.

\bibitem{fan1997comments}
J.~Fan, ``Comments on wavelets in statistics: A review by a. antoniadis,''
  \emph{Journal of the Italian Statistical Society}, vol.~6, no.~2, pp.
  131--138, 1997.

\bibitem{fan2001variable}
J.~Fan and R.~Li, ``Variable selection via nonconcave penalized likelihood and
  its oracle properties,'' \emph{Journal of the American Statistical
  Association}, vol.~96, no. 456, pp. 1348--1360, 2001.

\bibitem{zhang2010nearly}
C.-H. Zhang, ``Nearly unbiased variable selection under minimax concave
  penalty,'' \emph{Annals of Statistics}, vol.~38, no.~2, pp. 894--942, 2010.

\bibitem{fan2011nonconcave}
J.~Fan and J.~Lv, ``Nonconcave penalized likelihood with np-dimensionality,''
  \emph{IEEE Transactions on Information Theory}, vol.~57, no.~8, pp.
  5467--5484, 2011.

\bibitem{schumaker2007spline}
L.~Schumaker, \emph{Spline functions: basic theory}.\hskip 1em plus 0.5em minus
  0.4em\relax Cambridge University Press, 2007.

\bibitem{li2012scad}
G.~Li, L.~Xue, and H.~Lian, ``Scad-penalised generalised additive models with
  non-polynomial dimensionality,'' \emph{Journal of Nonparametric Statistics},
  vol.~24, no.~3, pp. 681--697, 2012.

\bibitem{cl1973some}
C.~L. Mallows, ``Some comments on {C}p,'' \emph{Technometrics}, vol.~42, no.~1,
  pp. 87--94, 2000.

\bibitem{chen2008extended}
J.~Chen and Z.~Chen, ``Extended bayesian information criteria for model
  selection with large model spaces,'' \emph{Biometrika}, vol.~95, no.~3, pp.
  759--771, 2008.

\bibitem{konishi1996generalised}
S.~Konishi and G.~Kitagawa, ``Generalised information criteria in model
  selection,'' \emph{Biometrika}, vol.~83, no.~4, pp. 875--890, 1996.

\bibitem{hui2015tuning}
F.~K. Hui, D.~I. Warton, and S.~D. Foster, ``Tuning parameter selection for the
  adaptive lasso using eric,'' \emph{Journal of the American Statistical
  Association}, vol. 110, no. 509, pp. 262--269, 2015.

\bibitem{konishi2004bayesian}
S.~Konishi, T.~Ando, and S.~Imoto, ``Bayesian information criteria and
  smoothing parameter selection in radial basis function networks,''
  \emph{Biometrika}, vol.~91, no.~1, pp. 27--43, 2004.

\bibitem{mandal2023robust}
A.~Mandal, B.~H. Beyaztas, and S.~Bandyopadhyay, ``Robust density power
  divergence estimates for panel data models,'' \emph{Annals of the Institute
  of Statistical Mathematics}, vol.~75, pp. 773--798, 2023.

\bibitem{hyvarinen2005estimation}
A.~Hyv{\"a}rinen, ``Estimation of non-normalized statistical models by score
  matching.'' \emph{Journal of Machine Learning Research}, vol.~6, no.~4, pp.
  695--709, 2005.

\bibitem{sugasawa2021selection}
S.~Sugasawa and S.~Yonekura, ``On selection criteria for the tuning parameter
  in robust divergence,'' \emph{Entropy}, vol.~23, no.~9, p. 1147, 2021.

\bibitem{fan2011nonparametric}
J.~Fan, Y.~Feng, and R.~Song, ``Nonparametric independence screening in sparse
  ultra-high-dimensional additive models,'' \emph{Journal of the American
  Statistical Association}, vol. 106, no. 494, pp. 544--557, 2011.

\bibitem{wang2011estimation}
L.~Wang, X.~Liu, H.~Liang, and R.~J. Carroll, ``Estimation and variable
  selection for generalized additive partial linear models,'' \emph{Annals of
  Statistics}, vol.~39, no.~4, pp. 1827--1851, 2011.

\bibitem{zhong2020forward}
W.~Zhong, S.~Duan, and L.~Zhu, ``Forward additive regression for
  ultrahigh-dimensional nonparametric additive models,'' \emph{Statistica
  Sinica}, vol.~30, no.~1, pp. 175--192, 2020.

\bibitem{kuchibhotla2022moving}
A.~K. Kuchibhotla and A.~Chakrabortty, ``Moving beyond sub-gaussianity in
  high-dimensional statistics: Applications in covariance estimation and linear
  regression,'' \emph{Information and Inference: A Journal of the IMA},
  vol.~11, no.~4, pp. 1389--1456, 2022.

\bibitem{wellner2017bennett}
J.~A. Wellner, ``The bennett-orlicz norm,'' \emph{Sankhya A}, vol.~79, no.~2,
  pp. 355--383, 2017.

\bibitem{devore1993constructive}
R.~A. DeVore and G.~G. Lorentz, \emph{Constructive approximation}.\hskip 1em
  plus 0.5em minus 0.4em\relax Springer Science \& Business Media, 1993, vol.
  303.

\bibitem{hampel1974influence}
F.~R. Hampel, ``The influence curve and its role in robust estimation,''
  \emph{Journal of the American Statistical Association}, vol.~69, no. 346, pp.
  383--393, 1974.

\bibitem{ollerer2015influence}
V.~{\"O}llerer, C.~Croux, and A.~Alfons, ``The influence function of penalized
  regression estimators,'' \emph{Statistics}, vol.~49, no.~4, pp. 741--765,
  2015.

\bibitem{ghosh2013robust}
A.~Ghosh and A.~Basu, ``Robust estimation for independent non-homogeneous
  observations using density power divergence with applications to linear
  regression,'' \emph{Electronic Journal of Statistics}, vol.~7, pp.
  2420--2456, 2013.

\bibitem{wang2007robust}
H.~Wang, G.~Li, and G.~Jiang, ``Robust regression shrinkage and consistent
  variable selection through the lad-lasso,'' \emph{Journal of Business \&
  Economic Statistics}, vol.~25, no.~3, pp. 347--355, 2007.

\bibitem{mgcvbook}
S.~Wood, \emph{Generalized Additive Models: An Introduction with R},
  2nd~ed.\hskip 1em plus 0.5em minus 0.4em\relax Chapman and Hall/CRC, 2017.

\bibitem{Rlanguage}
\BIBentryALTinterwordspacing
{R Core Team}, \emph{R: A Language and Environment for Statistical Computing},
  R Foundation for Statistical Computing, Vienna, Austria, 2022. [Online].
  Available: \url{https://www.R-project.org/}
\BIBentrySTDinterwordspacing

\bibitem{fdapackage}
\BIBentryALTinterwordspacing
J.~O. Ramsay, S.~Graves, and G.~Hooker, \emph{fda: Functional Data Analysis},
  2022, r package version 6.0.5. [Online]. Available:
  \url{https://CRAN.R-project.org/package=fda}
\BIBentrySTDinterwordspacing

\bibitem{shankavaram2007transcript}
U.~T. Shankavaram, W.~C. Reinhold, S.~Nishizuka, S.~Major, D.~Morita, K.~K.
  Chary, M.~A. Reimers, U.~Scherf, A.~Kahn, D.~Dolginow \emph{et~al.},
  ``Transcript and protein expression profiles of the nci-60 cancer cell panel:
  an integromic microarray study,'' \emph{Molecular Cancer Therapeutics},
  vol.~6, no.~3, pp. 820--832, 2007.

\bibitem{cortez2007data}
P.~Cortez and A.~d. J.~R. Morais, ``A data mining approach to predict forest
  fires using meteorological data,'' \emph{Associa{\c{c}}{\~a}o Portuguesa para
  a Intelig{\^e}ncia Artificial (APPIA)}, pp. 512--523, 2007.

\bibitem{lee2011sparse}
D.~Lee, W.~Lee, Y.~Lee, and Y.~Pawitan, ``Sparse partial least-squares
  regression and its applications to high-throughput data analysis,''
  \emph{Chemometrics and Intelligent Laboratory Systems}, vol. 109, no.~1, pp.
  1--8, 2011.

\bibitem{alfons2013sparse}
A.~Alfons, C.~Croux, and S.~Gelper, ``Sparse least trimmed squares regression
  for analyzing high-dimensional large data sets,'' \emph{Annals of Applied
  Statistics}, vol.~7, no.~1, pp. 226--248, 2013.

\bibitem{oshima1996oncogenic}
R.~G. Oshima, H.~Baribault, and C.~Caul{\'\i}n, ``Oncogenic regulation and
  function of keratins 8 and 18,'' \emph{Cancer and Metastasis Reviews},
  vol.~15, no.~4, pp. 445--471, 1996.

\bibitem{robusthd}
A.~Alfons, ``{robustHD}: An {R} package for robust regression with
  high-dimensional data,'' \emph{Journal of Open Source Software}, vol.~6,
  no.~67, p. 3786, 2021.

\bibitem{chen2018error}
Z.~Chen, J.~Fan, and R.~Li, ``Error variance estimation in
  ultrahigh-dimensional additive models,'' \emph{Journal of the American
  Statistical Association}, vol. 113, no. 521, pp. 315--327, 2018.

\bibitem{jones2001comparison}
M.~Jones, N.~L. Hjort, I.~R. Harris, and A.~Basu, ``A comparison of related
  density-based minimum divergence estimators,'' \emph{Biometrika}, vol.~88,
  no.~3, pp. 865--873, 2001.

\bibitem{maji2019robust}
A.~Maji, A.~Ghosh, A.~Basu, and L.~Pardo, ``Robust statistical inference based
  on the c-divergence family,'' \emph{Annals of the Institute of Statistical
  Mathematics}, vol.~71, pp. 1289--1322, 2019.

\bibitem{sara2008hdglm}
\BIBentryALTinterwordspacing
S.~A. van~de Geer, ``{High-dimensional generalized linear models and the
  lasso},'' \emph{Annals of Statistics}, vol.~36, no.~2, pp. 614--645, 2008.
  [Online]. Available: \url{https://doi.org/10.1214/009053607000000929}
\BIBentrySTDinterwordspacing

\end{thebibliography}

\newpage
\newgeometry{left=2cm,bottom=2cm,right=2cm}
\onecolumn

\setcounter{page}{1}

 \centerline{\large\bf Robust Variable Selection in High-dimensional}
\vspace{.25cm}
 \centerline{\large\bf  Nonparametric Additive model}
\vspace{2pt}
\vspace{.25cm}

\centerline{Suneel Babu Chatla and Abhijit Mandal} 

\vspace{.4cm}
 \centerline{\it The University of Texas at El Paso, Texas}
\vspace{.55cm}
 \centerline{\bf Supplementary Material}
\vspace{.55cm}

\setcounter{section}{0}
\setcounter{equation}{0}
\def\theequation{S\arabic{section}.\arabic{equation}}
\def\thesection{S\arabic{section}}

\fontsize{12}{14pt plus.8pt minus .6pt}\selectfont


\noindent
The supplementary material contains the proofs of Lemma \ref{lem:maximal-weibull} and Theorems \ref{thm:oracle} and \ref{thm:consistency}.  In the following lemma, we derive the maximal inequality for the product terms involving spline function and error. The proof requires results involving the Generalized Bernstein-Orlicz (GBO) norm. For completeness, we define it here. The function $\Psi_{\alpha, L}$ is defined in \cite{kuchibhotla2022moving}, with its inverse expressed as
\begin{align*}
    \Psi_{\alpha,L}^{-1} &:= \sqrt{\log(1+t)} + L (\log(1+t))^{1/\alpha} \text{  for all  } t \ge 0,
\end{align*}
for fixed $\alpha>0$ and $L \ge 0$. The GBO norm of a random variable $X$ is then
\begin{align*}
    \Vert X \Vert_{\Psi_{\alpha,L}} &:= \inf\{ \eta>0 ~:~ \mathbb{E}[\Psi_{\alpha,L}(|X|/\eta)] \le 1 \}.
\end{align*}

\begin{lemma} \label{lem:maximal-weibull}
Suppose assumptions \ref{as:1-function-class}--\ref{as:5-error} hold. Let
\begin{align*}
    T_{kj} &= n^{-1/2}m_n^{1/2} \sum_{i=1}^n B_{kj}(X_{ij}) \epsilon_{i}, \qquad \qquad 1 \le k \le m_n, 1 \le j \le p,
\end{align*}
and  $T_n=\underset{1 \le j \le p, 1 \le k \le m_n}{\max} |T_{kj}|$. Let $\alpha^*=\max\{0, \frac{\alpha-1}{\alpha}\}$. When $m_n\log(pm_n)/n \rightarrow 0$ as $n \rightarrow \infty$, we obtain the following bound:
\begin{align*}
         \mathbb{E}\left[ T_n \right]  =  \left( \sqrt{\log(pm_n)}    + n^{-1/2+\alpha^*}m_n^{1/2-\alpha^*}  (\log(pm_n))^{1/\alpha} \right) O(1).
\end{align*}
    
\end{lemma}
\begin{proof}
    From Remark A.1 in \cite{kuchibhotla2022moving}, conditional on $X_{ij}$'s, we obtain that
    \begin{align}
        \mathbb{E} &\left[ \underset{1 \le j \le p, 1 \le k \le m_n}{\max} |T_{kj}| \mid \{ X_{ij}, 1 \le j \le p, 1 \le i \le n\} \right] \nonumber \\& \qquad \le n^{-1/2}m_n^{1/2}\underset{1 \le j \le p, 1 \le k \le m_n}{\max} \Vert T_{kj} \Vert_{\Psi_{\alpha,L}} C_{\alpha} \left\{ \sqrt{\log(pm_n)} + L (\log(pm_n))^{1/\alpha} \right\}, \label{eqn:lem-max-max0}
    \end{align}
    for some constant $C_{\alpha}$ depending on $\alpha$. Similarly, the application of Theorem 3.1 in \cite{kuchibhotla2022moving} yields that
    \begin{align}
        \Vert T_{kj} \Vert_{\Psi_{\alpha,L}} &\le 2 e C_1(\alpha)  \Vert b \Vert_2, \label{eqn:lem-max-wsum} 
    \end{align}
    where $C_1(\alpha)$ is some constant involving $\alpha$ for which the explicit expression is given in Theorem 3.1 of \cite{kuchibhotla2022moving}, $\Vert b \Vert_2 =  \left(\sum_{i=1}^n B_{kj}^2(X_{ij}) \Vert \epsilon_i \Vert_{\psi_{\alpha}}^2 \right)^{1/2}$, and
    \begin{align*}
        L(\alpha) &= \frac{4^{1/\alpha}}{\sqrt{2}\Vert b \Vert_2} \times \begin{cases} \Vert b\Vert_{\infty}, & \text{if  } \alpha<1, \\
        4e \Vert b \Vert_{\varpi}/C_1(\alpha), & \text{if  } \alpha \ge 1, \end{cases}
    \end{align*}
    with $\varpi$ is the Holder conjugate satisfying $1/\alpha+1/\varpi=1$. Therefore, combination of (\ref{eqn:lem-max-max0}) and (\ref{eqn:lem-max-wsum}) yields

    \begin{align}
        \mathbb{E} &\left[ \underset{1 \le j \le p, 1 \le k \le m_n}{\max} |T_{kj}| \mid \{ X_{ij}, 1 \le j \le p, 1 \le i \le n\} \right] \nonumber \\& \qquad \le 2e C_{\alpha} C_1(\alpha) n^{-1/2}m_n^{1/2} \underset{1 \le j \le p, 1 \le k \le m_n}{\max} \Vert b \Vert_2   \sqrt{\log(pm_n)} \nonumber\\ &
        \qquad \qquad + n^{-1/2}m_n^{1/2}\begin{cases}  D_{1}(\alpha) \underset{1 \le j \le p, 1 \le k \le m_n}{\max} \Vert b \Vert_{\infty} (\log(pm_n))^{1/\alpha} & \text{ if } \alpha < 1 , \\  D_2(\alpha) \underset{1 \le j \le p, 1 \le k \le m_n}{\max} \Vert b \Vert_{\varpi} (\log(pm_n))^{1/\alpha} & \text{ if } \alpha \ge 1, \end{cases} \label{eqn:lem-max-max1}
    \end{align}
    where $D_1(\alpha)= \sqrt{2}e C_1(\alpha) C_{\alpha} 4^{1/\alpha}$ and $D_2(\alpha)= \sqrt{2}C_{\alpha} e^2 4^{1/\alpha+1}$.

    Let $s_{n1}=\underset{1 \le j \le p, 1 \le k \le m_n}{\max} \Vert b \Vert_2$, $s_{n2}= \underset{1 \le j \le p, 1 \le k \le m_n}{\max} \Vert b \Vert_{\infty}$, and $s_{n3}=\underset{1 \le j \le p, 1 \le k \le m_n}{\max} \Vert b \Vert_{\varpi}$. It follows that
\begin{align}
        \mathbb{E} &\left[ \underset{1 \le j \le p, 1 \le k \le m_n}{\max} |T_{kj}| \right] \nonumber \\& \qquad \le 2e C_{\alpha} C_1(\alpha) n^{-1/2}m_n^{1/2} \mathbb{E} \left( s_{n1} \right)  \sqrt{\log(pm_n)} \nonumber\\ &
        \qquad \qquad + n^{-1/2}m_n^{1/2} \begin{cases}  D_{1}(\alpha) \mathbb{E}\left( s_{n2} \right) (\log(pm_n))^{1/\alpha} & \text{ if } \alpha < 1 , \\  D_2(\alpha) \mathbb{E}\left( s_{n3} \right) (\log(pm_n))^{1/\alpha} & \text{ if } \alpha \ge 1. \end{cases} \label{eqn:lem-max-max9}
    \end{align}
    We now derive the upper bound for the $\mathbb{E}(s_{n3})$.  From Theorem 4.2 of Chapter 5 in \cite{devore1993constructive}, we obtain that
    \begin{align}
        \sum_{i=1}^n \mathbb{E}\left\{ B_{jk}^{\varpi}(X_{ij}) - \mathbb{E}B_{jk}^{\varpi}(X_{ij}) \right\}^2 \le C_3 n m_n^{-1}, \qquad \vert B_{jk}(X_{ij}) \vert \le 2,
        \label{eqn:lem-max-beta-1}
    \end{align}
    and 
    \begin{align}
        \underset{1 \le j \le p, 1 \le k \le m_n}{\max} \sum_{i=1}^n \mathbb{E}B_{jk}^{\varpi}(X_{ij}) \le C_4 n m_n^{-1}.
        \label{eqn:lem-max-beta-2}
    \end{align}
By Lemma A.1 of \cite{sara2008hdglm} and (\ref{eqn:lem-max-beta-1}) we obtain
\begin{align}
    \mathbb{E} &\left( \underset{1 \le j\le p, 1 \le k \le m_n}{\max} \left\vert \sum_{i=1}^n \mathbb{E}\left\{ B_{jk}^{\varpi}(X_{ij}) - \mathbb{E}B_{jk}^{\varpi}(X_{ij}) \right\} \right\vert \right) \nonumber \\ & \qquad \qquad \le \sqrt{2 C_3 n m_n^{-1} \log(2pm_n)} + 2^{\varpi} \log(2pm_n).
    \label{eqn:lem-max-beta-11}
\end{align}
Now, an application of Triangle inequality using (\ref{eqn:lem-max-beta-11}) and (\ref{eqn:lem-max-beta-2}) gives
\begin{align*}
    \mathbb{E}(s_{n3}^{\varpi}) &\le \sqrt{2 C_3 n m_n^{-1} \log(2pm_n)} + 2^{\varpi} \log(2pm_n)+ C_4 n m_n^{-1}.
\end{align*}
Therefore, using Lyapunov's inequality, we obtain that, for $1 < \varpi < \infty$,
\begin{align}
    \mathbb{E}(s_{n3}) &\le  \left(\sqrt{2 C_3 n m_n^{-1} \log(2pm_n)} + 2^{\varpi} \log(2pm_n)+ C_4 n m_n^{-1} \right)^{1/\varpi}.
    \label{eqn:lem-max-sn3}
\end{align}
Letting $\varpi=2$ in (\ref{eqn:lem-max-sn3}) yields the corresponding bound for $\mathbb{E}(s_{n1})$. By (\ref{eqn:lem-max-beta-1}), we can obtain that $\mathbb{E}(s_{n2}) \le C_5$. Consequently,
\begin{align}
        & \mathbb{E}\left[ T_n \right]  \le 2e C_{\alpha} C_1(\alpha) n^{-1/2}m_n^{1/2} \Big(\sqrt{2 C_3 n m_n^{-1} \log(2pm_n)}  \nonumber \\ &+ 4 \log(2pm_n)+ C_4 n m_n^{-1} \Big)^{1/2}  \sqrt{\log(pm_n)} + n^{-1/2}m_n^{1/2} \nonumber \\ &\times \begin{cases} C_5 D_{1}(\alpha)  (\log(pm_n))^{1/\alpha} & \text{ if } \alpha < 1 , \\  D_2(\alpha) \bigg(\sqrt{2 C_3 n m_n^{-1} \log(2pm_n)} \\  + 2^{\varpi} \log(2pm_n)+ C_4 n m_n^{-1} \bigg)^{1/\varpi} (\log(pm_n))^{1/\alpha} & \text{ if } \alpha \ge 1. \end{cases} \label{eqn:lem-max-max}
    \end{align}

Let $\alpha^*=\max\{0, \frac{\alpha-1}{\alpha}\}$. When $m_n\log(pm_n)/n \rightarrow 0$ as $n \rightarrow \infty$, the above bound simplifies to the following bound:
\begin{align*}
         \mathbb{E}\left[ T_n \right]  =  \left( \sqrt{\log(pm_n)}    + n^{-1/2+\alpha^*}m_n^{1/2-\alpha^*}  (\log(pm_n))^{1/\alpha} \right) O(1).
\end{align*}
    
\end{proof}

\begin{proof}[Proof of Theorem \ref{thm:oracle}]
Note that the oracle estimator $\widehat{\beta}^0= (\sqrt{m_n}\widehat{\mu}_0^0, \widehat{\beta}_1^0, \ldots, \widehat{\beta}_q^0)^T \in \mathbb{R}^{qm_n+1}$, minimizes the DPD loss function 
\begin{align}
\ell_n(\beta) := \frac{1}{n} \sum_{i=1}^n V_i(\beta; \sigma^2, \nu), \label{eqn:th1-ln} 
\end{align}
 where $V_i(\cdot)$ is defined in (\ref{eqn:vi}). We remark that penalization is not required for the oracle model as the true components are known. Using Taylor expansion, we have
\begin{align*}
    \frac{\partial \ell_n(\beta)}{\partial \beta} \bigg\vert_{\beta=\widehat{\beta}^0} &= \frac{\partial \ell_n(\beta)}{\partial \beta} + \frac{\partial^2 \ell_n(\beta)}{\partial \beta \partial \beta^T} \bigg\vert_{\beta=\overline{\beta}} ( \widehat{\beta}^0- \beta),
\end{align*}
where $\overline{\beta}=t\widehat{\beta}^0 + (1-t)\beta$, $t\in [0,1]$. Therefore,
\begin{align*}
    \widehat{\beta}^0- \beta &= -\left(\frac{\partial^2 \ell_n(\beta)}{\partial \beta \partial \beta^T} \bigg\vert_{\beta=\overline{\beta}} \right)^{-1} \frac{\partial \ell_n(\beta)}{\partial \beta}.
\end{align*}
Let $Z_i^0=(1/\sqrt{m_n}, B_{11}(X_{i1}), \ldots, B_{m_nq}(X_{iq}))^T$ be the corresponding spline basis for the first $q$ variables. By straightforward calculations, we have
\begin{align}
    \frac{\partial \ell_n(\beta)}{\partial \beta} &= -\frac{1+\nu}{n\sigma^2} \sum_{i=1}^n f_i^{\nu} \times (Y_i-Z_i^{0^T}\beta)Z_i^{0}, \label{eqn:oracle-1d}
\end{align}
and 
\begin{align}
    \frac{\partial^2 \ell_n(\beta)}{\partial \beta \partial \beta^T} \bigg\vert_{\beta=\overline{\beta}} &= -\frac{1+\nu}{n\sigma^2} \sum_{i=1}^n \left\{\frac{\nu}{\sigma^2} f_i^{\nu} \times (Y_i-Z_i^{0^T}\overline{\beta})^2Z_i^{0}Z_i^{0^T} - f_i^{\nu} Z_i^{0}Z_i^{0^T} \right\}. \label{eqn:oracle-2d}
\end{align}
First, we find the bound for (\ref{eqn:oracle-1d}). Let $\delta_i=\sum_{j=1}^q g_j(X_i)-g_{nj}(X_i)$. Observe that
\begin{align*}
    &\left\Vert -\frac{1+\nu}{n\sigma^2} \sum_{i=1}^n f_i^{\nu} \times (Y_i-Z_i^{0^T}\beta)Z_i^{0} \right\Vert = \left\Vert -\frac{1+\nu}{n\sigma^2} \sum_{i=1}^n f_i^{\nu} \times (\delta_i+\epsilon_i)Z_i^{0} \right\Vert   \\
    & \qquad \qquad \le \left\Vert -\frac{1+\nu}{n\sigma^2} \sum_{i=1}^n f_i^{\nu} \delta_i Z_i^{0} \right\Vert + \left\Vert -\frac{1+\nu}{n\sigma^2} \sum_{i=1}^n f_i^{\nu}\epsilon_i Z_i^{0} \right\Vert.
\end{align*}
From Fact 1 in (\ref{eqn:fact1}) and Lemma 1 in \cite{huang2010variable}, we have $\underset{i}{\max} ~ \delta_i  \le C_{12} q m_n^{-d}$ for some constant $C_{12}$. Using the fact that $\underset{i}{\max}~ f_i^{\nu} \le 1$, we obtain
\begin{align}
    & \left\Vert -\frac{1+\nu}{n\sigma^2} \sum_{i=1}^n f_i^{\nu}(Y_i-Z_i^{0^T}\beta^0)Z_i^{0} \right\Vert \nonumber \\ &\le C_{12} q m_n^{-d} \frac{1+\nu}{\sigma^2} \left\Vert \frac{1}{n}\sum_{i=1}^n Z_{i}^{0} \right\Vert +  \frac{1+\nu}{\sigma^2}  \left\Vert \frac{1}{n} \sum_{i=1}^n \epsilon_i Z_{i}^{0} \right\Vert \nonumber\\
    &=C_{12} q m_n^{-d-1/2} \frac{1+\nu}{\sigma^2} + \frac{1+\nu}{\sigma^2}  \left\Vert \frac{1}{n} \sum_{i=1}^n \epsilon_i Z_{i}^{0} \right\Vert \label{eqn:oracle-1d-a}
\end{align}
where the second step follows because of $n^{-1}\sum_{i=1}^n B_{kj}(X_{ij})=0$, for $ 1\le k \le m_n, 1\le j \le q$. Consider
\begin{align}
   & \left\Vert \frac{1}{n} \sum_{i=1}^n \epsilon_i Z_{i}^{0} \right\Vert^2  \nonumber \\ &= \sum_{j=1}^q\sum_{k=1}^{m_n} \left( \frac{1}{n} \sum_{i=1}^n \epsilon_i B_{kj}(X_{ij})\right)^2 + \left( \frac{1}{n\sqrt{m_n}} \sum_{i=1}^n \epsilon_i \right)^2 \nonumber\\
   &\le \underset{1 \le j \le q, 1\le k \le m_n}{\max} \frac{q}{n}  \left( \frac{m_n^{1/2}}{n^{1/2}} \sum_{i=1}^n \epsilon_i B_{kj}(X_{ij})\right)^2 + \left( \frac{1}{n\sqrt{m_n}} \sum_{i=1}^n \epsilon_i \right)^2 \nonumber\\
   &= O_p(1) (q/n) \left( \log(qm_n)    + n^{-1+2\alpha^*}m_n^{1-2\alpha^*}  (\log(qm_n))^{2/\alpha} \right) +O_p(1/nm_n), \label{eqn:oracle-1d-b}
\end{align}
where the first term in the last step follows from Lemma \ref{lem:maximal-weibull} and the second term follows from Condition \ref{as:5-error}. Combination of (\ref{eqn:oracle-1d}), (\ref{eqn:oracle-1d-a}), and (\ref{eqn:oracle-1d-b}) yields
\begin{align}
    & \left\Vert \frac{\partial \ell_n(\beta)}{\partial \beta}  \right\Vert \nonumber \\ &= O\left( m_n^{-d-1/2} \right) + O_p\left(  \sqrt{\frac{\log(qm_n)+ n^{-1+2\alpha^*}m_n^{1-2\alpha^*}  (\log(qm_n))^{2/\alpha}}{n} + \frac{1}{nm_n}} \right). \label{eqn:oracle-1d-final}
\end{align}
We now show that,
\begin{align}
    \left\Vert \left(-\frac{1+\nu}{n\sigma^2} \sum_{i=1}^n  \left\{\frac{\nu}{\sigma^2} f_i^{\nu} \times (Y_i-Z_i^{0^T}\overline{\beta})^2Z_i^{0}Z_i^{0^T} - f_i^{\nu} Z_i^{0}Z_i^{0^T} \right\} \right)^{-1} \right\Vert 
    & = O\left( m_n \right). \label{eqn:oracle-2d-final}
\end{align}
From (\ref{eqn:oracle-2d}), we write
\begin{align*}
     \sum_{i=1}^n \frac{1+\nu}{n}  \left\{\frac{1}{\sigma^2} f_i^{\nu}\times \left[1- \frac{\nu(Y_i-Z_i^{0^T}\overline{\beta})^2}{\sigma^2}\right]\right\}Z_i^{0}Z_i^{0^T}. 
\end{align*}
First, observe that $\left[1- \frac{\nu(Y_i-Z_i^{0^T}\overline{\beta})^2}{\sigma^2} \right] \le \exp\{- \frac{\nu(Y_i-Z_i^{0^T}\overline{\beta})^2}{2\sigma^2}\} \le 1$ for $i=1,\ldots,n$. Consequently $f_i^{\nu} \times \left[1- \frac{\nu(Y_i-Z_i^{0^T}\overline{\beta})^2}{\sigma^2}\right] \le 1/\sigma\sqrt{2\pi}$. Further, since $\exp\{\frac{\nu(Y_i-Z_i^{0^T}\overline{\beta})^2}{2\sigma^2}\} \ge \left[1- \frac{\nu(Y_i-Z_i^{0^T}\overline{\beta})^2}{\sigma^2} \right]$, we have the following lower bound
\begin{align*}
    f_i^{\nu}\left[1- \frac{\nu(Y_i-Z_i^{0^T}\overline{\beta})^2}{\sigma^2}\right] \ge \left[1- \frac{\nu(Y_i-Z_i^{0^T}\overline{\beta})^2}{\sigma^2} \right]^2/\sigma \sqrt{2\pi},
\end{align*}
which is always positive and takes zero  when $(Y_i-Z_i^{0^T}\overline{\beta})^2=\sigma^2/\nu$. Together, we have,
\begin{align}
    1/\sigma\sqrt{2\pi} \ge f_i^{\nu} \times \left[1- \frac{\nu(Y_i-Z_i^{0^T}\overline{\beta})^2}{\sigma^2}\right] \ge \left[1- \frac{\nu(Y_i-Z_i^{0^T}\overline{\beta})^2}{\sigma^2} \right]^2/\sigma \sqrt{2\pi}, \label{eqn:2d-wts-positive}
\end{align}
which concludes that the weights are positive and bounded. Therefore, by Lemma 3 in \cite{huang2010variable}, we prove (\ref{eqn:oracle-2d-final}). Consequently, the result (\ref{eqn:oracle-final}) follows from (\ref{eqn:oracle-1d-final}) and (\ref{eqn:oracle-2d-final}).
\end{proof}


\begin{proof}[Proof of Theorem \ref{thm:consistency}]
The main idea of the proof is similar to Theorem 3 in \cite{fan2011nonconcave}. However, the details are more involved due to the group penalty and the basis functions. Without loss of generality, we assume the first $q$ components, $g_j$, $j=1,\ldots,q$, are nonzero. Let $\beta=(\beta^{(1)^T}, \beta^{(2)^T})^T  \in \mathbb{R}^{pm_n+1}$ with $\beta^{(1)}=(\mu, \beta_1^T, \ldots, \beta_q^T)^T$ and $\beta^{(2)}= 0$.

\vspace{2em}

\emph{Step1: Consistency in the $q-$dimensional space: }

Let $Z_i^{(1)}=(1, B_{11}(X_{i1}), \ldots, B_{m_nq}(X_{iq}))^T$ and  \\
$Z_i^{(2)}= (B_{1q+1}(X_{iq+1}), \ldots, B_{m_np}(X_{ip}))^T$ be the basis functions corresponding to first $q$ nonzero functions (including intercept) and $p-q$ zero functions, respectively. First, we constrain the likelihood $L_{\nu}$ to $qm_n+1$ dimensional subspace as the following:
\begin{align}
    \overline{L}_{\nu}(\delta) &= -\ell_n(\delta) - \sum_{j=1}^q P_{\lambda}(\Vert \delta_j\Vert_2),
\end{align}
where $\delta=(1, \delta_1^T, \ldots, \delta_q^T)^T$ with $\delta_j=(\delta_{1j}, \ldots, \delta_{m_nj})^T$ and $\ell_n$ is defined in (\ref{eqn:th1-ln}). Note that we take negative signs in $\overline{L}_{\nu}(\delta)$  so that we now maximize the likelihood instead of minimizing $L_{\nu}$ in (\ref{eqn:plike}).
 We show that there exists a local maximizer $\widehat{\beta}^{(1)}$ of $\overline{L}_{\nu}(\delta)$ such that $\Vert \widehat{\beta}^{(1)} - \beta^{(1)} \Vert= O_p(\gamma_n)$. Define an event 
\begin{align}
    H_n &= \left\{ \underset{\delta \in \partial N_{\tau}}{\max} \overline{L}_{\nu}(\delta) < \overline{L}_{\nu}(\beta^{(1)}) \right\},
\end{align}
where $ \partial N_{\tau}$ is the boundary of the closed set
\begin{align*}
    N_{\tau} &= \{\delta \in \mathbb{R}^{qm_n+1}: \Vert \delta - \beta^{(1)}  \Vert \le \gamma_n \tau\},
\end{align*}
and $\tau \in (0, \infty)$ and rate $\gamma_n$. Note that on event, $H_n$ there exists a local maximizer $\widehat{\beta}^{(1)}$ of $\overline{L}_{\nu}(\delta)$ in $N_{\tau}$. Therefore, it is sufficient to show that $P(H_n)$ approaches to 1 as $n \rightarrow \infty$.

By Taylor expansion, for any $\delta \in N_{\tau}$, we have
\begin{align}
      \overline{L}_{\nu}(\delta) -\overline{L}_{\nu}(\beta^{(1)})  
    &=  (\delta- \beta^{(1)})^T V - \frac{1}{2} (\delta- \beta^{(1)})^T D (\delta- \beta^{(1)}), \label{eqn:th2-taylor}
\end{align}
where 
\begin{align*}
    V := \nabla \overline{L}_{\nu}(\beta^{(1)}) &= \frac{1+\nu}{n\sigma^2} \sum_{i=1}^n f_i^{\nu}\cdot (Y_i-Z_i^{(1)^T}\beta^{(1)}) Z_i^{(1)}  - \overline{P}_{\lambda}(\beta^{(1)}),
\end{align*}
with $\overline{P}_{\lambda}(\beta^{(1)})=( P'_{\lambda}(\Vert \beta_1\Vert_2) (D_1 \beta_1)^T/\Vert \beta_1\Vert_2, \ldots, P'_{\lambda}(\Vert \beta_q\Vert_2) (D_q \beta_q)^T/\Vert \beta_q\Vert_2)^T$ where $P'_{\lambda}$ is the derivative of the penalty function, and 
\begin{align}
    D &:= -\nabla^2 \overline{L}_{\nu}(\beta^{*}) =\sum_{i=1}^n \frac{1+\nu}{n}  \left\{\frac{1}{\sigma^2} f_i^{\nu}\left[1- \frac{\nu(Y_i-Z_i^{(1)^T}\beta^*)^2}{\sigma^2}\right]\right\}Z_i^{(1)}Z_i^{(1)^T} \nonumber \\ 
    &+ \text{diag}\left\{ \frac{P'_{\lambda}(\Vert \beta_j^* \Vert_2)}{\Vert \beta_j^* \Vert_2} D_j + \left[ \frac{P^{''}_{\lambda}(\Vert \beta_j^* \Vert_2)}{\Vert \beta_j^* \Vert_2^2} -\frac{P'_{\lambda}(\Vert \beta_j^* \Vert_2)}{\Vert \beta_j^* \Vert_2^3} \right] D_j \beta_j^* \beta_j^{*^T} D_j \right\}, \label{eqn:th2-taylor-2t}
\end{align}
where $\beta^*$ is on the line segment joining $\beta^{(1)}$ and $\delta$.
We note that the matrix $D_j$ is positive definite, and its  eigenvalues  are of order $m_n^{-1}$. 
For any $\delta \in \partial N_{\tau}$, we have $\Vert \delta - \beta^{(1)} \Vert=\gamma_n \tau$ and $\beta^{*} \in N_{\tau}$.
By doing calculations analogous to (\ref{eqn:oracle-2d-final}), we obtain from (\ref{eqn:th2-taylor-2t}) that 
\begin{align}
    \lambda_{\text{min}}(D) &\ge C_{21}( m_n^{-1} - \lambda  \kappa_0 m_n^{-2} ).  
\end{align}
Thus by (\ref{eqn:th2-taylor}), we have
\begin{align*}
      \underset{\delta \in \partial N_{\tau}}{\max} \overline{L}_{\nu}(\delta) -\overline{L}_{\nu}(\beta^{(1)})  &\le  \gamma_n \tau [\Vert V \Vert- C_{21} \gamma_n \tau ( m_n^{-1} - \lambda \kappa_0 m_n^{-2})],
\end{align*}
which, together with Markov's inequality, gives 
\begin{align*}
    P(H_n) &\ge P\left( \Vert V \Vert^2 < C_{21}^2 \gamma_n^2 \tau^2 (m_n^{-1} - \lambda  \kappa_0 m_n^{-2})^2  \right) \\
    &\ge 1- \frac{\mathbb{E}\Vert V \Vert^2}{C_{21}^2 \gamma_n^2 \tau^2 (m_n^{-1}  - \lambda \kappa_0 m_n^{-2})^2}.
\end{align*}
Calculations analogous to (\ref{eqn:oracle-1d-final}) in Theorem \ref{thm:oracle} yield that
\begin{align*}
    \mathbb{E}\Vert V \Vert^2 &\le \mathbb{E} \left\Vert \frac{1+\nu}{n\sigma^2} \sum_{i=1}^n f_i^{\nu}\cdot (Y_i-Z_i^{(1)^T}\beta^{(1)}) Z_i^{(1)} \right\Vert^2 + \Vert \overline{P}_{\lambda}(\beta^{(1)}) \Vert^2 \\
     &\le O\left( m_n^{-2d-1} + \frac{\log(pm_n)}{n} + \frac{m_n^{1-2\alpha^*} (\log(p m_n))^{2/\alpha} }{n^{2-2\alpha^*}} + \frac{1}{nm_n} \right) + O(q \lambda^2 m_n^{-2}).
\end{align*}
Consequently,
\begin{align*}
    P(H_n) \ge 1- \frac{O\left( m_n^{-2d+1} + \frac{m_n^2\log(pm_n)}{n} +\frac{m_n^{3-2\alpha^*} (\log(p m_n))^{2/\alpha} }{n^{2-2\alpha^*}} + \frac{m_n}{n} + q \lambda^2 \right)}{C_{21}^2 \gamma_n^2 \tau^2 (1- \lambda \kappa_0m_n^{-1})^2}.
\end{align*}
By choosing $\gamma^2_n = m_n^{-2d+1} + \frac{m_n^2\log(pm_n)}{n} +\frac{m_n^{3-2\alpha^*} (\log(p m_n))^{2/\alpha} }{n^{2-2\alpha^*}}+ \frac{m_n}{n} + q \lambda^2$ and based on $\lambda \kappa_0 = o(1)$ (condition \ref{as:6-penalty}), we have
\begin{align*}
    P(H_n) &\ge 1-o(\tau^{-2}).
\end{align*}
It proves
\begin{align*}
    \sum_{j=1}^q \Vert \widehat{\beta}_j -\beta_j \Vert^2 = &  O_p\left(\frac{m_n^2\log(pm_n)}{n} +\frac{m_n^{3-2\alpha^*} (\log(p m_n))^{2/\alpha} }{n^{2-2\alpha^*}} + \frac{m_n}{n} \right)  \\
    &+ O\left( m_n^{-2d+1} +  q \lambda^2 \right).
\end{align*}
 

\noindent
\emph{Step 2: Sparsity:}

From Step 1 we have $\widehat{\beta}^{(1)} \in \mathbb{R}^{qm_n+1}$ is a local maximizer of $\overline{L}_{\nu}$ on $N_{\tau}$. We now prove that  $\widehat{\beta}= (\widehat{\beta}^{(1)^T}, 0^T)^T$ is indeed a maximizer of $-L_{\nu}$ on the space $\mathbb{R}^{pm_n+1}$.  Let $\xi=(1, \xi_1^T, \ldots,\xi_p^T)^T= \sum_{i=1}^n f_i^{\nu}\cdot (Y_i-Z_i^T\beta) Z_i$. Let $\widehat{\beta}_{S_0}= \widehat{\beta}^{(1)}$ and $\widehat{\beta}_{S_0^c}= \widehat{\beta}^{(2)}=0$. Consider the event,
\begin{align*}
    \mathcal{E}_2 &= \left\{ \Vert \xi_{S_0^c} \Vert_{\infty} \le u_n \right\}.
\end{align*}
 consider 
 \begin{align*}
     \Vert \xi_{S_0^c} \Vert_{\infty} &= \Vert  \sum_{i=1}^n f_i^{\nu}\cdot (Y_i-Z_i^T\beta) Z_i^{(2)} \Vert_{\infty} \\
     &\le n^{1/2}m_n^{-1/2} \cdot \underset{1 \le j \le p, 1 \le k \le m_n }{\max}\vert n^{-1/2}m_n^{1/2} \sum_{i=1}^n (Y_i-Z_i^{T}\beta) B_{kj}(X_{ij}) \vert\\
     &=   n^{1/2}m_n^{-1/2} \cdot T_{n},
 \end{align*}
 where $T_{n}=\underset{1 \le j \le p, 1 \le k \le m_n }{\max}\vert n^{-1/2}m_n^{1/2} \sum_{i=1}^n (Y_i-Z_i^{T}\beta) B_{kj}(X_{ij}) \vert$. We have,
 \begin{align*}
     P(\mathcal{E}_2) &\ge  P \left(T_n  \le \frac{u_n}{(n/m_n)^{1/2}} \right) \\
     &\ge 1 -  \frac{\mathbb{E}T_n}{u_n}(n/m_n)^{1/2} \\
     &\ge 1- \frac{\left( \sqrt{\log(pm_n)}    + n^{-1/2+\alpha^*}m_n^{1/2-\alpha^*}  (\log(pm_n))^{1/\alpha} \right)(n/m_n)^{1/2}}{u_n},
 \end{align*}
 which follows from Lemma 2 in \cite{huang2010variable} after ignoring the small spline approximation error. Choosing $u_n= n/m_n$, we obtain
 \begin{align*}
     P(\mathcal{E}_2) &\ge 1- \left(  \sqrt{\frac{m_n \log(pm_n)}{n}} + \frac{m_n^{1-\alpha^*}(\log(pm_n))^{1/\alpha}}{n^{1-\alpha^*}} \right)   \rightarrow 1, 
 \end{align*}
 under the assumption that $m_n\log(pm_n)/n \rightarrow 0$ as $n\rightarrow \infty$.
 Following Theorem \ref{thm:consistency} in \cite{fan2011nonconcave}, it is sufficient to show that, 
 \begin{align*}
     \Vert W \Vert_{\infty} &:=  n^{-1} \Vert \sum_{i=1}^n f_i^{\nu}\cdot (Y_i- Z_i^{T}(-\beta+ \widehat{\beta}+ \beta)) Z_i^{(2)} \Vert_{\infty}\\
     &\le n^{-1} \left[ \Vert \xi_{S_0^c} \Vert_{\infty} +  \Vert \sum_{i=1}^n f_i^{\nu}\cdot (Z_i^{T}(\widehat{\beta}- \beta)) Z_i^{(2)} \Vert_{\infty} \right] \\
     &\le o(1) + n^{-1} \Vert \sum_{i=1}^n Z_i^{T}(\widehat{\beta}- \beta) Z_i^{(2)} \Vert \\
     &\le o(1) + n^{-1} \Vert \sum_{i=1}^n Z_i^{(1)^T}(\widehat{\beta}^{(1)}- \beta^{(1)}) Z_i^{(2)} \Vert\\
     &= o(1) +  O(m_n^{-1}) \Vert (\widehat{\beta}^{(1)}- \beta^{(1)}) \Vert \\
     &=o(1),
 \end{align*}
 which concludes the proof.
\end{proof}

\end{document}